%% file: 0_main.tex
\documentclass[sigconf]{acmart}
\AtBeginDocument{%
  \providecommand\BibTeX{{%
    \normalfont B\kern-0.5em{\scshape i\kern-0.25em b}\kern-0.8em\TeX}}}

\usepackage{amsfonts}       % blackboard math symbols
\usepackage{nicefrac}       % compact symbols for 1/2, etc.
\usepackage{microtype}      % microtypography
\usepackage{xcolor}         % colors
\usepackage{graphicx}
\usepackage{amsmath}
\usepackage{amsthm}
\usepackage{bm}
\usepackage{bbm}
\usepackage{algorithm}
\usepackage{algorithmic}
\usepackage{subcaption}
\DeclareMathOperator*{\argmax}{argmax}

\begin{document}

%%
%% The "title" command has an optional parameter,
%% allowing the author to define a "short title" to be used in page headers.
\title[Hierarchical Conversational Preference Elicitation with Bandit Feedback]{Hierarchical Conversational Preference Elicitation with Bandit Feedback}

\author{Jinhang Zuo}
\affiliation{%
 \institution{Carnegie Mellon University}
 \city{Pittsburgh}
 \state{Pennsylvania}
 \country{USA}}
\email{jzuo@andrew.cmu.edu}

\author{Songwen Hu}
\affiliation{%
 \institution{Shanghai Jiao Tong University}
 \city{Shanghai}
 \country{China}}
\email{hsw828@sjtu.edu.cn}	

\author{Tong Yu}
\affiliation{%
 \institution{Adobe Research}
 \city{San Jose}
 \state{California}
 \country{USA}}
\email{tyu@adobe.com}
 
\author{Shuai Li}
\authornote{Corresponding author.}
\affiliation{%
 \institution{Shanghai Jiao Tong University}
 \city{Shanghai}
 \country{China}}
\email{shuaili8@sjtu.edu.cn}

\author{Handong Zhao}
\affiliation{%
 \institution{Adobe Research}
 \city{San Jose}
 \state{California}
 \country{USA}}
\email{hazhao@adobe.com}

\author{Carlee Joe-Wong}
\affiliation{%
 \institution{Carnegie Mellon University}
 \city{Pittsburgh}
 \state{Pennsylvania}
 \country{USA}}
\email{cjoewong@andrew.cmu.edu}

\begin{abstract}
The recent advances of conversational recommendations provide a promising way to efficiently elicit users' preferences via conversational interactions. To achieve this, the recommender system conducts conversations with users, asking their preferences for different items or item categories. Most existing conversational recommender systems for cold-start users utilize a multi-armed bandit framework to learn users' preference in an online manner. However, they rely on a pre-defined conversation frequency for asking about item categories instead of individual items, which may incur excessive conversational interactions that hurt user experience. To enable more flexible questioning about key-terms, we formulate a new conversational bandit problem that allows the recommender system to choose either a key-term or an item to recommend at each round and explicitly models the rewards of these actions. This motivates us to handle a new exploration-exploitation (EE) trade-off between key-term asking and item recommendation,  which requires us to accurately model the relationship between key-term and item rewards. We conduct a survey and analyze a real-world dataset to find that, unlike assumptions made in prior works, key-term rewards are mainly affected by rewards of representative items. We propose two bandit algorithms, Hier-UCB and Hier-LinUCB, that leverage this observed relationship and the hierarchical structure between key-terms and items to efficiently learn which items to recommend. We theoretically prove that our algorithm can reduce the regret bound's dependency on the total number of items from previous work. We validate our proposed algorithms and regret bound on both synthetic and real-world data.
\end{abstract}

%%
%% The code below is generated by the tool at http://dl.acm.org/ccs.cfm.
%% Please copy and paste the code instead of the example below.
%%


%%
%% Keywords. The author(s) should pick words that accurately describe
%% the work being presented. Separate the keywords with commas.

%% A "teaser" image appears between the author and affiliation
%% information and the body of the document, and typically spans the
%% page.
% \begin{teaserfigure}
%   \includegraphics[width=\textwidth]{sampleteaser}
%   \caption{Seattle Mariners at Spring Training, 2010.}
%   \Description{Enjoying the baseball game from the third-base
%   seats. Ichiro Suzuki preparing to bat.}
%   \label{fig:teaser}
% \end{teaserfigure}

%%
%% This command processes the author and affiliation and title
%% information and builds the first part of the formatted document.
\maketitle
\input{1_intro.tex}
\input{5_related}
\input{2_model.tex}
\input{3_algorithm.tex}
\input{4_experiment}
\input{6_conclusion}
\bibliography{references}
\bibliographystyle{unsrt}
% \clearpage
% \input{9_appendix}
\end{document}

%% file: 1_intro.tex
\section{Introduction}
Recommender systems have attracted much attention in a variety of areas over the past decades, e.g., to steer users towards movies they will enjoy~\cite{deldjoo2018audio} or suggest applications that they may find useful~\cite{liu2015personalized}. Most recommender systems utilize large amounts of historical data to learn user behaviour and preferences.
However, for cold-start users who interact with the system for a short amount of time, there is not enough historical data to reliably learn user preferences, and thus the system has to interact with these users in an online manner in order to make informed recommendations.
Successfully learning user preferences from this interaction, however, is challenging, as most users will not tolerate extensive exploration of potential items \cite{zhang2020conversational,liu2018transferable}.
%Instead, recommender systems should only recommend items that are likely of interest to the user, which risks missing some items that the user may enjoy if they are not related to the first few items that are recommended.
The recent advances of conversational recommender systems (CRS) propose a new way to efficiently elicit users’ preferences~\cite{dalton2018vote,sardella2019approach,christakopoulou2016towards,zhang2020conversational,xie2021comparison,zhao2022knowledge}. The system is allowed to conduct conversations with users to get richer feedback information, which enables better preference elicitation and alleviates extensive explorations. These conversations are expected to be less burdensome for users than receiving irrelevant recommendations, allowing for less onerous exploration of item recommendations~\cite{christakopoulou2016towards}.

\textbf{Research challenges.} Although CRSs have achieved success in many fields (e.g., movie~\cite{dalton2018vote} and restaurant~\cite{sardella2019approach} recommendation), current conversation mechanisms still face challenges. Most mechanisms utilize a multi-armed bandit framework to decide what the agent should ask a user, based on the results of previous conversational interactions. Existing conversational bandits~\cite{christakopoulou2016towards,zhang2020conversational,xie2021comparison} rely on a pre-defined conversation frequency. Specifically, at some pre-determined conversation rounds, the recommender agent asks users questions about ``key-terms,'' i.e., their preferences for categories of items that could be recommended, and then receives feedback that it can use to infer users' item preferences. %\carlee{This information is then used to recommend items to users in subsequent interactions.}
%These key-term asking conversations are expected to be less annoying to users than random item recommendations, but we still want to minimize the number of conversational interactions. 
\begin{figure}[t]
    \centering
    \includegraphics[width=0.45\textwidth]{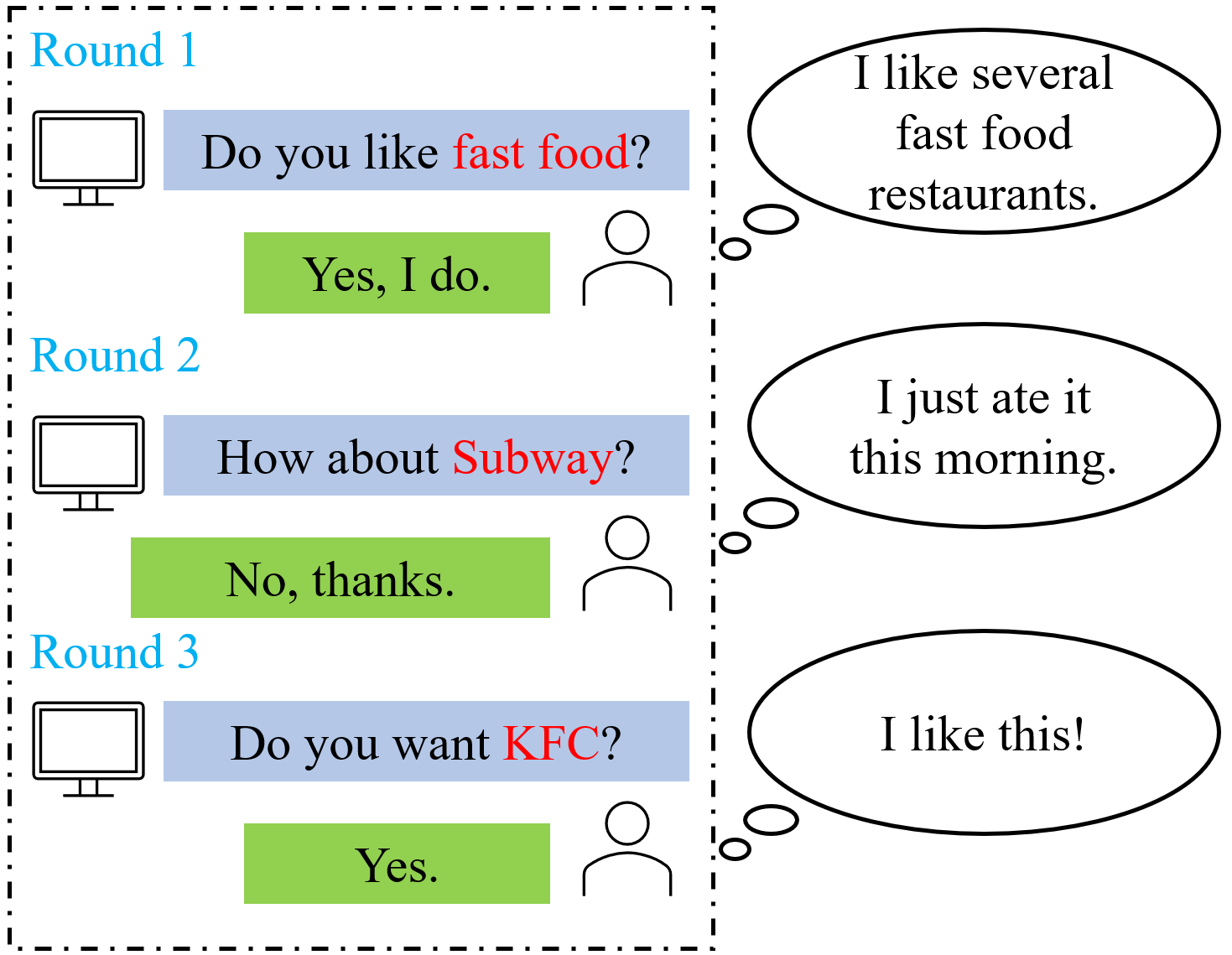}
    \caption{An example showing the importance of accurately modeling the relationship between the user preferences for key-terms and items. The user preference for `fast food' is affected by some fast food restaurants (\emph{i.e.}, `KFC' in this example), rather than all fast food restaurants. The user's rejection of `Subway' does not mean that the agent should always be penalized for recommending the key-term `fast food'. In fact, the user does like `fast food' in this case. However, after round 2, existing approaches \cite{zhang2020conversational,wu2021clustering,li2021seamlessly} will penalize the agent for recommending the key-term `fast food' in round 1.}
    \label{fig:RestaurantRec}
\end{figure}
However, even though conversational interactions may be more tolerable to users than poor item recommendations~\cite{christakopoulou2016towards}, we should still limit the number of required interactions for finding an optimal item~\cite{li2021seamlessly}.
Thus, for users whose preferences are easy to learn, we want to have fewer interactions, e.g., for users with clear preferences between key-term categories, more frequent key-term questions upfront may allow us to quickly identify a good key-term category containing items that the user likes. Prior work~\cite{li2021seamlessly} takes a step in this direction by unifying key-term questions and item recommendation in the same action space, but their solutions do not have theoretical guarantees and do not explicitly consider the cost of key-term questions, which can still incur burdensome conversational interactions. %Also, they only provide algorithms without theoretical guarantees.

A natural approach to enable more flexible key-term question interactions is to treat each such interaction as an alternative to item recommendation. These interactions can then be associated with a reward indicating how much the user likes this category of items, just as item recommendations are associated with rewards that represent how much a user likes the item. We can then model the recommender agent's actions as a sequence of actions, each drawn from an action space that includes both conversational key-term interactions and item recommendations. The goal is to choose a sequence of actions, informed by the results from prior actions, that maximizes the cumulative reward, i.e., user satisfaction. In doing so, we implicitly limit the number of required interactions before finding high-reward items (i.e., ones the user likes). To achieve this, \textbf{we formulate a new conversational bandit problem} that allows the agent to choose either a key-term or an item to recommend at each round and explicitly models the rewards of both actions.

This conversational bandit setting introduces a \emph{new exploration-exploitation (EE) trade-off between key-term questions and item recommendation}, which can not be handled by previous work. Balancing such an EE trade-off is challenging as it is affected by the complicated relationship between item rewards and key-term rewards.
One example of this relationship in conversational restaurant recommendations is illustrated by Figure~\ref{fig:RestaurantRec}.
The recommender agent first recommends key-term "fast food" to the user, which satisfies the user as he likes fast food restaurants such as KFC, Chipotle, and Wendy's. The agent then recommends "Subway", but the user rejects that recommendation since he just ate it shortly before and does not want to eat it again. The agent finally recommends another fast food restaurant "KFC", which is accepted by the user, and the interaction ends. In this example, the user's preference for "fast food" is mainly determined by his enjoyment of "KFC", but not his rejection of "Subway". %In other words, the reward of key-term "fast food" is directly related with the reward of the ``good'' item "KFC".
Intuitively, users' preferences for one key-term (i.e., the rewards associated with that key-term) are primarily affected by their preferences for the best (i.e., most preferred) items within that key-term category.
%mainly affected by their preferences of some representative items within that key-term, but not all items associated with that key-term.
Another example where the key-term reward is determined by the rewards of the best items can be found in image retrieval~\cite{jia2020personalized}: given a specific query from the user, we may show the user different clusters (key-terms) of images. If a cluster contains any images the user is looking for, he should be satisfied with that cluster. Thus, the user's satisfaction of a cluster is decided by whether it contains the desired images.

To verify our above intuition of the relationship between item and key-term rewards, we conduct a survey (as detailed in Section~\ref{sec:key-term}) asking participants to rate restaurants (items) and corresponding categories (key-terms).
We find that \emph{the item and key-term reward relationship defined by previous work is inaccurate for many categories}. Specifically, \cite{zhang2020conversational,xie2021comparison} consider the reward of a key-term to be the average of the rewards of all associated items. 
However, we find this average is usually smaller than the true key-term ratings from participants,
%We verify the intuition from the examples that 
%Instead, the true ratings are mainly determined by 
which are closer to the ratings of the most highly rated associated items. 
These findings are consistent with an analysis of restaurant ratings from the Yelp dataset (Section~\ref{sec:key-term}).

\textbf{Our contributions.} Inspired by the above observation that key-term rewards are mainly affected by the rewards of representative items, we leverage this relationship and the hierarchical structure between key-terms and items to develop more efficient recommendation algorithms. 
% Such item and key-term reward relationship makes it possible to take advantages of the hierarchical structure between key-terms and items for efficient online learning. 
Since maximizing the cumulative reward is achieved by quickly identifying the item that most appeals to the user, we can accelerate our search by first identifying the most appealing key-term (which, from our results above, will contain the most appealing item).
% However, deciding when to switch between key-terms and items is still the core and challenging problem. Switching to items too fast risks identifying the wrong optimal key-term, making it impossible to find the optimal item, while over-exploration of key-terms negates any runtime benefits from exploring key-terms instead of items.
To do so, we propose two algorithms, Hier-UCB (upper confidence bound) and Hier-LinUCB, for stochastic and contextual conversational bandit settings, respectively. We prove the theoretical regret bound of the proposed algorithm reduces the dependency on the total number of items, from $|\mathcal{A}|$ in previous work~\cite{zhang2020conversational} to $|\mathcal{K}| + |\mathcal{A}_{k^*}|$, where $|\mathcal{A}|$ is the size of the item pool, $|\mathcal{K}|$ is the size of the key-term pool, and $|\mathcal{A}_{k^*}|$ is the number of constituent items of the best key-term.
Our regret analysis is also different from the previous hierarchical Bayesian bandit works~\cite{wan2021metadata,hong2021hierarchical,kveton2021meta,hong2022deep,basu2021no}: with a careful treatment of the regrets caused by choosing sub-optimal key-terms or items, we provide a stronger frequentist regret bound that, unlike the previous Bayes regret bounds, is agnostic to the quality of prior information.
Experimental results on a synthetic dataset show that Hier-UCB outperforms both traditional UCB and Hier-LinUCB with a moderated item pool. We run Hier-UCB with three real-world datasets with large item pools and find it outperforms baseline algorithms from previous works.

In summary, our \textbf{research contributions} are as follows:
\begin{itemize}
	\item To enable flexible key-term question frequencies for preference elicitation, we \emph{formulate a conversational bandit problem} that explicitly models the exploration-exploitation (EE) trade-off between key-term questions and item recommendation.
	\item We \emph{study the relationship between item and key-term rewards based on a survey and analysis of a real-world dataset}. We have an essential observation for our algorithm design that the key-term rewards are mainly determined by the rewards of representative items.
	\item We propose \emph{Hier-UCB and Hier-LinUCB}, efficient bandit algorithms that leverage the observed relationship and the hierarchical structure between attributes and items. We prove that Hier-UCB reduces the theoretical regret bound's dependency on the total number of items in prior work.
	\item We \emph{validate our proposed algorithms and regret bounds} on both synthetic and real-world data.
\end{itemize}
The rest of the paper is organized as follows. We formulate the new conversational bandit problem and study the item-key-term reward relationship in Section~\ref{sec:model}. In Section~\ref{sec:algorithm}, we propose Hier-UCB and Hier-LinUCB and analyze their regrets. The experimental setup and results are presented in Section~\ref{sec:exp}. We discuss some related works in Section~\ref{sec:related} and conclude the paper in Section~\ref{sec:conclusion}.

%% file: 5_related.tex
\section{Related Work}\label{sec:related}
\subsection{Conversational Preference Elicitation}
The recent advances of conversational recommender systems (CRS) provide a promising way to efficiently elicit users’ preferences~\cite{dalton2018vote,sardella2019approach,christakopoulou2016towards,zhang2020conversational,xie2021comparison,christakopoulou2018q,priyogi2019preference}. 
By developing advanced models and efficient algorithms in CRS, we can promisingly derive more information about the user preferences and understand their preferences more accurately~\cite{zhang2018towards,chen2019towards,zou2020towards,li2018towards,liu2020towards,yu2019visual,zhang2019text,zhou2020improving,xu2022neural,zhou2020towards}.
%\tong{
%With the recent advances in deep learning and natural language processing (NLP), considerable efforts have been dedicated to integrating the conversation systems and recommender systems to ask more information about the user preferences in conversation~\cite{zhang2018towards,chen2019towards,zou2020towards} and to better understand the user preferences in user responses~\cite{chen2019towards,li2018towards,liu2020towards,yu2019visual,zhou2020improving,zhou2020towards}. 
%By formulating the recommendation procedure as a multi-step decision-making process, other researchers adopt reinforcement learning (RL) methods~\cite{deng2021unified,lei2020estimation,lei2020interactive,sun2018conversational,xu2021adapting,yu2019visual} to decide what to ask in conversations and when to conduct the conversations.
%}
To determine when to have the conversations and what to ask during conversations, previous works formulate the task as a multi-step decision-making process and propose reinforcement learning based approaches \cite{deng2021unified,lei2020estimation,lei2020interactive,sun2018conversational,xu2021adapting}.
The most related works of this paper are a series of studies that utilize a multi-armed bandit framework for preference elicitation from cold-start users in an online manner. 
Christakopoulou et al.~\cite{christakopoulou2016towards} first introduce the use of bandit-based strategies for preference elicitation into CRS. 
Zhang et al.~\cite{zhang2020conversational} formulate the conversational contextual bandit problem that integrates linear contextual bandits and CRS. They propose the ConUCB algorithm, one of our baselines, which achieves a smaller theoretical regret upper bound than the conventional contextual bandit algorithm LinUCB~\cite{li2010contextual}. 
A follow-up~\cite{li2021seamlessly} models key-term asking and item recommendation in an unified framework and proposes a bandit solution based on Thompson Sampling.
Xie et al.~\cite{xie2021comparison} introduce a comparison-based CRS with relative feedback to the comparisons of key-terms. 
However, none of these works explicitly model the cost (i.e., incurred regret) of the conversations that ask users about their key-term preferences in the conversational bandit framework. Thus, they may let users suffer from excessive conversational interactions. To the best of our knowledge, we are the first to formally model the regret of key-term asking and study the new EE trade-off between key-term asking and item recommendation, to enable more flexible conversations. We further revisit the relationship between key-term and item rewards and, based on our results, propose a new model for key-term rewards as determined by the rewards of representative items, instead of all items as assumed in prior work.

\subsection{Bandits with Hierarchy}
There is a branch of research that relates to our proposed bandit algorithms that take advantages of a hierarchy structure. Hierarchical Bayesian bandits have been studied in~\cite{wan2021metadata,hong2021hierarchical,kveton2021meta,hong2022deep,basu2021no}, with applications in meta-learning and multi-task learning.  A hierarchical Bayesian framework is proposed in \cite{wan2021metadata} to solve the metadata-based multi-task multi-armed bandit problem. Hong et al.~\cite{hong2021hierarchical} formulate a general hierarchical Bayesian bandit problem for solving multiple similar bandit tasks sequentially or concurrently, where each task is parameterized by task parameters drawn from a distribution parameterized by hyper-parameters. Hong et al.~\cite{hong2022deep} study the bandit problem with a deep hierarchical structure other than the traditional 2-level Bayesian hierarchy. 
However, all these works rely on the prior knowledge of arms and specific Bayesian hierarchical structures that may not be practical in real-world applications. Our proposed algorithms do not require any prior information and utilize the simple hierarchical structure between key-terms and items that is available in most applications. Moreover, we provide a stronger frequentist regret bound than the Bayes regret bounds in previous works, which is agnostic to the quality of priors. Our analysis for the frequentist regret is also different from that for the Bayes regret: without the hierarchical Bayesian structure, we need to carefully bound the regrets caused by the 4 cases in Section 3.3.

%% file: 2_model.tex
\section{Problem Formulation}\label{sec:model}
In this section, we will formulate conversational recommendation as a conversational bandit problem. We start by introducing the stochastic conversational bandit scenario where an agent can recommend either a key-term or an item to a user. 
We then discuss the limitations of the previous works' models of key-term rewards and propose a new key-term reward function, based on both human evaluation and data analysis. 
We also provide a contextual version of our new conversational bandit problem.

\subsection{Stochastic Conversational Bandit}\label{sec:stochastic}
We formulate a novel stochastic conversational bandit problem that allows the recommender agent to choose either key-term asking or item recommendation at each round. It explicitly models the regret caused by key-term asking, which is ignored by prior work~\cite{zhang2020conversational,li2021seamlessly,xie2021comparison}.
The goal of the agent is to maximize the cumulative reward in $T$ rounds through repeated key-term and item recommendations.

Consider a finite set of items (e.g., restaurants) denoted by $\mathcal{A}$ and a finite set of key-terms (e.g., restaurant categories) denoted by $\mathcal{K}$. The relationships between the arms and the key-terms are defined by a weighted bipartite graph $(\mathcal{A},\mathcal{K},W)$, where item $a\in\mathcal{A}$ is associated to a key-term $k\in\mathcal{K}$ with weight $W_{a,k} \ge 0$. Without loss of generality, we assume $\sum_{k\in\mathcal{K}}W_{a,k}=1$. More discussion of the weights can be found in Section~\ref{sec:key-term}.

At each round $t=1,\cdots,T$, the agent needs to choose an item $a_t\in\mathcal{A}$ or a key-term $k_t\in\mathcal{K}$ and shows it to the user. 
If the agent chooses an item, it will receive an item reward:
\begin{equation}
    r_{a_t,t} = \mu_{a_t} + \epsilon_{t},
\end{equation}
where $\mu_{a_t}$ is the expected reward of item $a_t$ and $\epsilon_{t}$ is a random variable representing the random noise. Similarly, if the agent chooses a key-term, it will receive a key-term reward:
\begin{equation}
    \tilde{r}_{k_t,t} = \tilde{\mu}_{k_t} + \tilde{\epsilon}_{t},
\end{equation}
where $\tilde{\mu}_{k_t}$ is the expected reward of key-term $k_t$ and $\tilde{\epsilon}_{t}$ is a random variable representing the random noise. For ease of presentation, if the agent chooses an item $a_t\in\mathcal{A}$, we assume it also chooses $k_t=\emptyset$; if the agent chooses a key-term $k_t\in\mathcal{K}$, we assume it also chooses $a_t=\emptyset$. With a slight abuse of notation, we let $r_{\emptyset,t} = \tilde{r}_{\emptyset,t} = 0$. We define the real obtained reward of the agent at round $t$ as:
\begin{equation}
    R_{a_t, k_t, t} = r_{a_t,t} + \tilde{r}_{k_t,t}.
\end{equation}
The goal of the agent is to maximize the expected cumulative rewards.
Let $\sum_{t=1}^{T}\mathbb{E}[R_{a^*_t,k^*_t,t}]$ denote the maximum expected cumulative reward in $T$ rounds, where $(a^*_t,k^*_t)$ is the optimal action at round $t$, i.e., $\mathbb{E}[R_{a^*_t,k^*_t,t}] \ge \mathbb{E}[R_{a,k,t}]$, $\forall a\in \mathcal{A}, b=\emptyset$ or $\forall k\in \mathcal{K}, a=\emptyset$. 
The goal of the agent is to maximize the expected cumulative reward or,
equivalently, to minimize the expected cumulative regret in $T$ rounds:
\begin{equation}\label{eq:regret}
    Reg(T) = \sum_{t=1}^T \left(\mathbb{E}[R_{a^*_t,k^*_t,t}] - \mathbb{E}[R_{a_t,k_t,t}]\right).
\end{equation}

\subsection{Key-term Reward Revisited}\label{sec:key-term}
\begin{figure}
    \centering
    \includegraphics[width=0.3\textwidth,trim={0 0.7cm 0 1.2cm},clip]{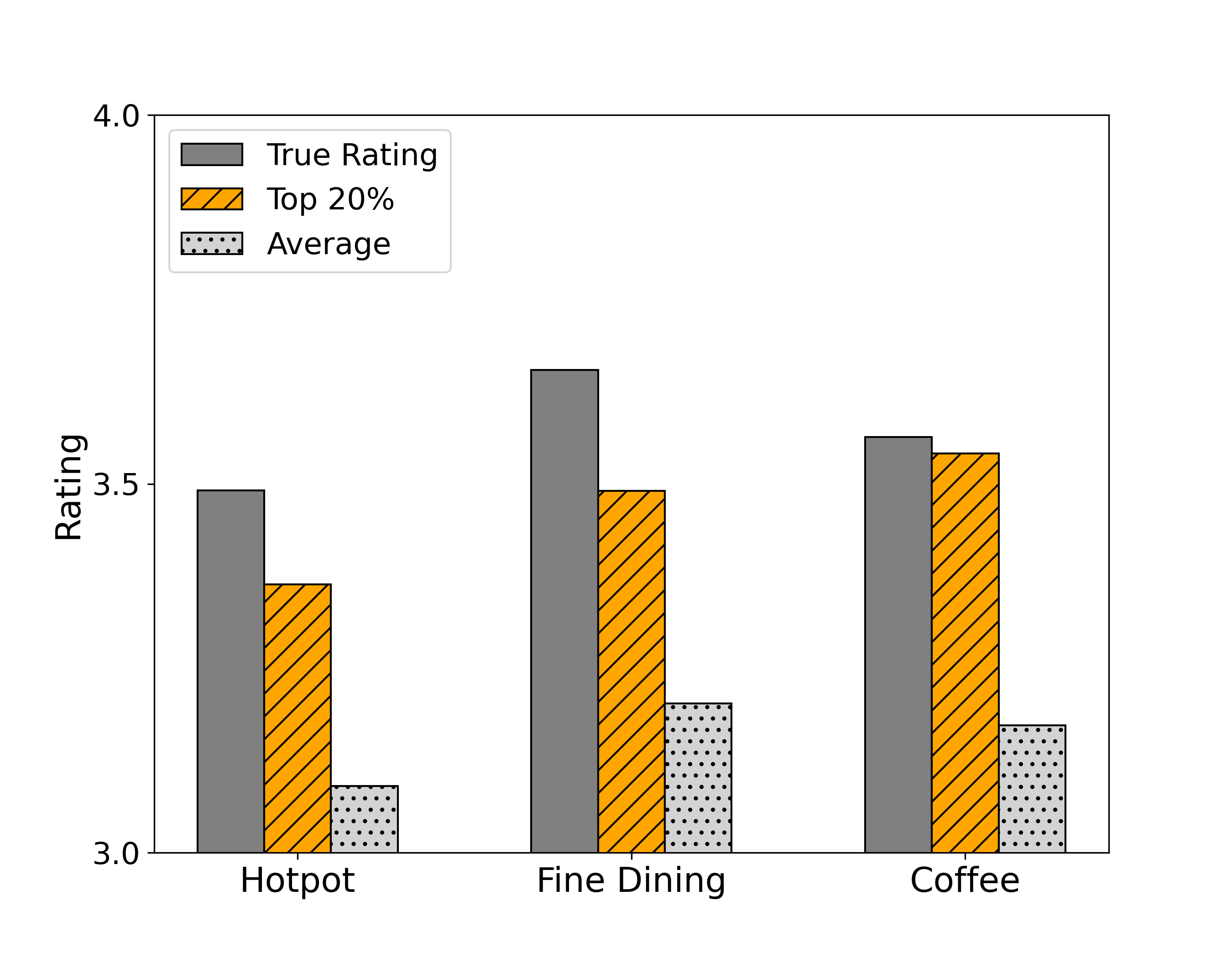}
    \caption{Restaurant category ratings from our survey are closer to the ratings of the top 20\% of items in each category than the average ratings of all items in each category.}
    \label{fig:he}
\end{figure}
\begin{figure}
    \centering
    \includegraphics[width=0.28\textwidth]{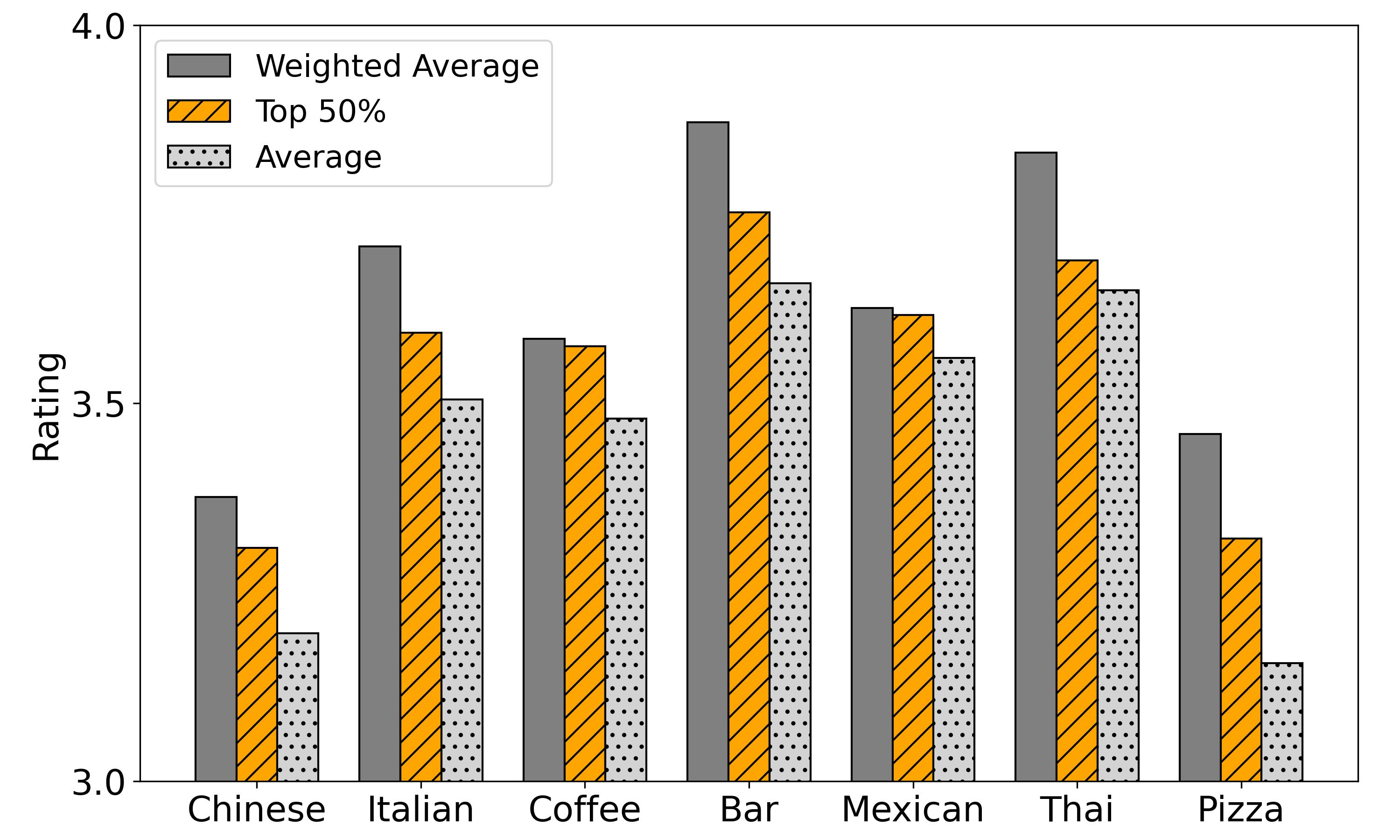}
    \caption{Restaurant category ratings from the Yelp dataset, showing similar category-item relationships as in Figure~\ref{fig:he}.}
    %. The observations are similar as the ones in our survey.}
    \label{fig:he_yelp}
\end{figure}

Since we allow the agent to choose from key-term or item recommendations, the relationship between key-term and item rewards plays an important role in deciding which action to take at each round. Previous work~\cite{zhang2020conversational} proposes that the reward of key-term $k$ is the weighted average of the rewards of all related items:
\begin{equation}
    \mathbb{E}[\tilde{r}_{k,t}] = \sum_{a\in\mathcal{A}} \frac{W_{a,k}}{\sum_{a'\in\mathcal{A}}W_{a',k}}\mathbb{E}[r_{a,t}],\,\,\, k\in\mathcal{K}.
\end{equation}
However, it is unrealistic for the recommender agent to know the exact weight $W_{a,k}$ between item $a$ and key-term $k$ in advance. It is also impractical to learn these weights online. 
As a result, previous work considers simplified binary weights for real applications, i.e., $W_{a,k}=1$ if item $a$ belongs to key-term $k$, otherwise $W_{a,k}=0$.
For example, if there are only 3 restaurants $a_1,a_2,a_3$ belonging to key-term $k_1$, then the expected reward of key-term $k_1$ is calculated as $\mathbb{E}[\tilde{r}_{k_1,t}]=\frac{\mathbb{E}[r_{a_1,t}] + \mathbb{E}[r_{a_2,t}] + \mathbb{E}[r_{a_3,t}]}{3}$. We call this key-term reward calculation method the "simple average", and we study its correctness from a survey and data analysis on a real dataset.

%\tong{}
We conduct the survey by asking the participant questions. Each participant is asked to first rate 3 restaurant categories (key-terms), "Hotpot", "Fine Dining", and "Coffee", and then 10 restaurants (items) randomly sampled from 50 restaurants for each category, on a scale of 1 to 5. 
The rating of category has an extra option "I have never tried", and participants who check this option will not be counted into the analysis of the corresponding category. In total, there are 55 participants.
The true ratings of different categories and the ratings calculated by a simple average are shown in Figure~\ref{fig:he}. We find the simple average ratings are always smaller 
% \carlee{do you have standard deviations across participants? it's hard to tell if the difference is significant otherwise.}\jinhang{how to explain that} 
than the true ratings, which suggests that simple average is not an accurate way to calculate key-term rewards. 
One potential explanation is that people's preferences for a restaurant category are mainly determined by their preferences for some representative restaurants, so the simple average method underestimates the key-term reward as it treats all associated restaurants in the same category equally.

To verify this idea, we propose a new key-term reward calculation method called "top-$\alpha$ average", which takes the average reward of the top $\alpha$ fraction of restaurants as the key-term reward. We choose $\alpha=20\%$ for better comparison with the simple average ratings, while $\alpha\in[0.1, 0.5]$ generally provides similar results. As shown in Figure~\ref{fig:he}, top-$20\%$ average ratings are much closer to the true ratings than simple average ratings.

Since the scale of the survey is relatively small, to further verify our findings, we run a {\bf data analysis on the Yelp dataset~\cite{yelp_dataset}}.
We choose $7$ restaurant categories with over $10,000$ reviews for each. We calculate the weighted average ratings of them, by taking the review counts of restaurants as the weights. This key-term reward can be interpreted as the average score of a review for a restaurant in that category. Notice that the recommender agent should not know the review counts before its recommendations, so weighted average ratings are only used for evaluation but not recommendation.
We compare weighted average ratings with simple average ratings and top-$50\%$ average ratings in Figure~\ref{fig:he_yelp}. Top-$50\%$ average ratings are closer to weighted average ratings than simple average ratings, which is consistent with our findings from our survey.

\subsection{Contextual Conversational Bandit}
We next introduce a contextual version of the new conversational bandit problem, which allows us to more easily accommodate large item pools. The agent plays with a user with unknown feature $\bm{\theta}^*\in \mathbb{R}^d$ for $T$ rounds. At each round $t$,  the agent observes contextual vector $\bm{x}_{a,t}$ for all items $a\in\mathcal{A}$ and contextual vector $\tilde{\bm{x}}_{k,t}$ for all key-terms $k\in\mathcal{K}$. It then needs to choose an item $a_t\in\mathcal{A}$ or a key-term $k_t\in\mathcal{K}$ and shows it to the user, as in the non-contextual problem. We follow LinUCB~\cite{li2010contextual} in modeling the reward as a linear function of the context.
If the agent chooses an item, it will receive an item reward:
\begin{equation}
    r_{a_t,t} = \bm{x}^T_{a_t,t}\bm{\theta}^* + \epsilon_{t},
\end{equation}
where $\bm{x}^T_{a_t,t}\bm{\theta}^*$ is the expected reward of item $a_t$ and $\epsilon_{t}$ is a random variable representing the random noise. Similarly, if the agent chooses a key-term, it will receive a key-term reward:
\begin{equation}
    \tilde{r}_{k_t,t} = \tilde{\bm{x}}^T_{k_t,t}\bm{\theta}^* + \tilde{\epsilon}_{t},
\end{equation}
where $\tilde{\mu}_{k_t}$ is the expected reward of key-term $k_t$ and $\tilde{\epsilon}_{t}$ is a random variable representing the random noise. The definitions of the obtained reward $R_{a_t, k_t, t}$ and the regret $Reg(T)$ are the same as those in Section~\ref{sec:stochastic}. The goal of the agent is still to minimize the expected regret $Reg(T)$, defined as in \eqref{eq:regret}.

%% file: 3_algorithm.tex
\section{Algorithm \& Theoretical Analysis}\label{sec:algorithm}
In this section, we propose two algorithms, Hier-UCB and Hier-LinUCB, for the stochastic and contextual conversational bandit settings, respectively. As discussed in Section~\ref{sec:key-term}, the reward of a key-term is mainly determined by the rewards of the representative items, and the proposed algorithms take advantages of this finding and the hierarchical structure between key-terms and items to achieve more efficient learning.

% "top-$\alpha$ average" is a good approximation of the true key-term reward. 
% For ease of presentation, we assume that the reward of key-term $k$ is equal to the discounted reward of the best associated item, i.e., $\mathbb{E}[\tilde{r}_{k,t}]= \lambda \cdot \max_{a}\left( W_{a,k}\mathbb{E}[r_{a,t}]\right)$, where $\lambda\le1$ is a discount factor. \carlee{do you need this assumption for the algorithms? If not, move to the regret analysis section} Notice that our algorithms can be easily extended to a more general case, which will be discussed in the later section.
% \jinhang{may change to the assumption that $\mathbb{E}[\tilde{r}_{k,t}] \ge \mathbb{E}[\tilde{r}_{k',t}]$ if $\mathbb{E}[r_{a^*_{k},t}] \ge \mathbb{E}[r_{a^*_{k'},t}]$ and compare to the categorized bandits assumption}.
% Though we assume the key-term reward $\mathbb{E}[\tilde{r}_{k,t}]= \lambda \cdot \max_{a}\left( W_{a,k}\mathbb{E}[r_{a,t}]\right)$, Hier-UCB can also work under a more general assumption that $a^*_{k^*} = \argmax_a \mathbb{E}[r_{a,t}]$, i.e., the best item in the key-term with the highest reward also has the highest reward among all items.

\subsection{Hier-UCB}
Inspired by the well-known upper-confidence bound (UCB) solution algorithm for traditional multi-armed bandit formulations, we propose a Hierarchical UCB (Hier-UCB) algorithm described in Algorithm~\ref{alg:Hier-UCB} for the stochastic setting of our conversational bandit problem. The algorithm maintains the empirical mean $\hat{\mu}_{a}, \tilde{\hat{\mu}}_{k}$ and a confidence radius $\rho_{a}, \tilde{\rho_{k}}$ for each item $a\in\mathcal{A}$ and key-term $k\in\mathcal{K}$.
The confidence radius is designed to be large if item $a$ or key-term $k$ is not chosen often ($T_{a}$ or $\tilde{T}_{k}$, which denote the number of times item $a$ or key-term $k$ has been recommended, is small) as is typical in UCB-based algorithms, so as to encourage exploration of rarely sampled arms or key-terms.

At each round, the algorithm first calculates key-term $\bar{k}^*$ with the highest UCB value among all key-terms, and then finds item $\bar{a}^*$ with the highest UCB value from all associated items of key-term $\bar{k}^*$ (line 6). It then decides whether to select an item or a key-term based on a switching condition (line 7): if the condition is not satisfied, which indicates the key-terms have not been sufficiently explored, it will
recommend key-term $\bar{k}^*$ (line 8) for one time, then recommend item $\bar{a}^*$ (line 12); otherwise, it will only recommend item $\bar{a}^*$ (line 12). Note that both key-term and item recommendations can incur regret.
Intuitively, the switching condition ensures that an item is chosen if a conservative estimate of the best item reward exceeds a generous estimate of the best key-term reward: in other words, if the best item's reward exceeds the best key-term's reward with high probability. The parameter $\gamma$ in the switching condition controls how easily the condition can be satisfied, and Theorem~\ref{thm:hierUCB} gives the appropriate values of $\gamma$.

\begin{algorithm}[tb]
 \caption{Hier-UCB}\label{alg:Hier-UCB}
 \begin{algorithmic}[1]
 \STATE \textbf{Input}: graph $(\mathcal{A},\mathcal{K}, W), \gamma$.
 \STATE \textbf{Init}: $T_a = 0, \tilde{T}_k = 0, \hat{\mu}_a = 1, \tilde{\hat{\mu}}_k =  1$.
 \WHILE{$t \le T$}
    \STATE For each item $a$, $\rho_a = \sqrt{\frac{3\ln t}{2 T_a}}$.
    \STATE For each key-term $k$, $\tilde{\rho}_{k} = \sqrt{\frac{3\ln t}{2\tilde{T}_{k}}}$.
    \STATE $\bar{k}^* = \argmax_{k\in\mathcal{K}} \tilde{\hat{\mu}}_k + \tilde{\rho}_{k}$, $\bar{a}^* = \argmax_{a\in\mathcal{A}} W_{a,\bar{k}^*}(\hat{\mu}_{a} + \rho_{a})$
    \IF {\NOT($\hat{\mu}_{\bar{a}^*} - \gamma \cdot \rho_{\bar{a}^*} \ge \tilde{\hat{\mu}}_{\bar{k}^*} + \gamma \cdot \tilde{\rho}_{\bar{k}^*}$)}
    \STATE Choose key-term $\bar{k}^*$, receive reward $\tilde{r}_{\bar{k}^*,t}$.
    \STATE Update $\tilde{T}_{\bar{k}^*}$, $\tilde{\hat{\mu}}_{\bar{k}^*}$: $\tilde{T}_{\bar{k}^*} = \tilde{T}_{\bar{k}^*}+1, \tilde{\hat{\mu}}_{\bar{k}^*} = \tilde{\hat{\mu}}_{\bar{k}^*} + (\tilde{r}_{\bar{k}^*,t}-\tilde{\hat{\mu}}_{\bar{k}^*}) / \tilde{T}_{\bar{k}^*}$.\STATE $t=t+1$.
    \ENDIF
    \STATE Choose item $\bar{a}^*$, receive reward $r_{\bar{a}^*,t}$.
    \STATE Update $T_{\bar{a}^*}$, $\hat{\mu}_{\bar{a}^*}$: $T_{\bar{a}^*} = T_{\bar{a}^*} + 1, \hat{\mu}_{\bar{a}^*} = \hat{\mu}_{\bar{a}^*} + (r_{\bar{a}^*,t}-\hat{\mu}_{\bar{a}^*}) / T_{\bar{a}^*}$.\STATE $t=t+1$.
 \ENDWHILE
\end{algorithmic} 
\end{algorithm}

\subsection{Hier-LinUCB}
We propose a Hierarchical LinUCB (Hier-LinUCB) algorithm, described in Algorithm~\ref{alg:Hier-LinUCB}, for the contextual setting. It is a modified version of Hier-UCB based on the LinUCB~\cite{li2010contextual} algorithm for contextual bandits.
Hier-LinUCB maintains user feature estimate $\bm{\theta}, \tilde{\bm{\theta}}$ and confidence radius $C_{a,t}, \tilde{C}_{k,t}$ for each item $a\in\mathcal{A}$ and key-term $k\in\mathcal{K}$. 
% \carlee{This is a bit confusing; does the $\theta$ estimate depend on each item or key-term? Or do you only have two estimates, one for items and one for key-terms? Why maintain separate estimates for the same thing?}
At each round, the algorithm first calculates the key-term $\bar{k}^*$ with the highest UCB value among all key-terms (line 7) and then finds the item $\bar{a}^*$ with the highest UCB value from all associated items of key-term $\bar{k}^*$ (line 8). It then decides whether to only explore on item $\bar{a}^*$ based on a switching condition (line 9): if the condition is false, which indicates the key-terms have not been sufficiently explored, it will recommend key-term $\bar{k}^*$ for one time (line 10), then recommend item $\bar{a}^*$ for one time (line 14).
Otherwise, it will only choose item $\bar{a}^*$ to recommend (line 14). Again, the switching condition controls when to recommend key-terms or items and is based on comparing a conservative estimate of the best item reward with a generous estimate of the key-term reward. The parameter $\gamma$ again controls how conservative or generous these estimates are.

\begin{algorithm}[tb]
 \caption{Hier-LinUCB}\label{alg:Hier-LinUCB}
 \begin{algorithmic}[1]
 \STATE \textbf{Input}: graph $(\mathcal{A},\mathcal{K}, W), \gamma$.
 \STATE \textbf{Init}: $\tilde{\bm{M}} = \bm{I}, \tilde{\bm{b}} = 0, \bm{M} = \bm{I}, \bm{b} = 0$.
 \WHILE{$t \le T$}
    \STATE $\tilde{\bm{\theta}} = \tilde{\bm{M}}^{-1}\tilde{\bm{b}}, \bm{\theta} = \bm{M}^{-1}\bm{b}$.
    \STATE For each item $a$, $C_{a,t} = \alpha_{t}||\tilde{\bm{x}}_{k,t}||_{\tilde{\bm{M}}^{-1}}$.
    \STATE For each key-term $k$, $\tilde{C}_{k,t} =\alpha_{t}||\tilde{\bm{x}}_{k,t}||_{\tilde{\bm{M}}^{-1}}$.
    \STATE $\bar{k}^* = \argmax_{k\in\mathcal{K}} \tilde{\bm{x}}_{k,t}^{T}\tilde{\bm{\theta}} + \tilde{C}_{k,t}$.
    \STATE $\bar{a}^* = \argmax_{a\in\mathcal{A}} W_{a,\bar{k}^*}\left(\bm{x}_{a,t}^{T}\bm{\theta} + C_{a,t}\right)$.
    \IF{\NOT($\bm{x}_{\bar{a}^*,t}^{T}\bm{\theta} - \gamma C_{\bar{a}^*,t} \ge \tilde{\bm{x}}_{\bar{k}^*,t}^{T}\tilde{\bm{\theta}} + \gamma \tilde{C}_{\bar{k}^*,t}$)}
        \STATE Choose key-term $\bar{k}^*$, receive reward $\tilde{r}_{\bar{k}^*}$.
        \STATE Update $\tilde{\bm{M}}$, $\tilde{\bm{b}}$: 
        $\tilde{\bm{M}} = \tilde{\bm{M}} + \tilde{\bm{x}}_{\bar{k}^*,t}\tilde{\bm{x}}_{\bar{k}^*,t}^{T}$, $\tilde{\bm{b}} = \tilde{\bm{b}} + \tilde{\bm{x}}_{\bar{k}^*,t}\tilde{r}_{\bar{k}^*,t}^{T}$.
        \STATE $t=t+1$.
    \ENDIF
    \STATE Choose item $\bar{a}^*$, receive reward $\bar{r}_{a^*,t}$.
    \STATE Update $\bm{M}$, $\bm{b}$: $\bm{M} = \bm{M} + \bm{x}_{\bar{a}^*,t}\bm{x}_{\bar{a}^*,t}^{T}$, $\bm{b} = \bm{b} + \bm{x}_{\bar{a}^*,t}r_{\bar{a}^*,t}^{T}$.
    \STATE $t=t+1$.
 \ENDWHILE
\end{algorithmic} 
\end{algorithm}

\subsection{Regret Analysis}\label{sec:regret}
% For ease of presentation, we assume that the reward of key-term $k$ is equal to the discounted reward of the best associated item, i.e., $\mathbb{E}[\tilde{r}_{k,t}]= \lambda \cdot \max_{a}\left( W_{a,k}\mathbb{E}[r_{a,t}]\right)$, where $\lambda\le1$ is a discount factor. \carlee{do you need this assumption for the algorithms? If not, move to the regret analysis section} Notice that our algorithms can be easily extended to a more general case, which will be discussed in the later section.
% \jinhang{may change to the assumption that $\mathbb{E}[\tilde{r}_{k,t}] \ge \mathbb{E}[\tilde{r}_{k',t}]$ if $\mathbb{E}[r_{a^*_{k},t}] \ge \mathbb{E}[r_{a^*_{k'},t}]$ and compare to the categorized bandits assumption}.
% Though we assume the key-term reward $\mathbb{E}[\tilde{r}_{k,t}]= \lambda \cdot \max_{a}\left( W_{a,k}\mathbb{E}[r_{a,t}]\right)$, Hier-UCB can also work under a more general assumption that $a^*_{k^*} = \argmax_a \mathbb{E}[r_{a,t}]$, i.e., the best item in the key-term with the highest reward also has the highest reward among all items.
In this section, we analyze the regret \eqref{eq:regret} of Hier-UCB in the stochastic conversational setting.
Let $\mathcal{A}_k = \{a \mid a\in\mathcal{A}, W_{a,k} \neq 0\}$ denote the set of all associated items for key-term $k$, $k^*$ denote the key-term with the highest expected reward, and $a_k^* = \argmax_{a\in\mathcal{A}_k} \mu_a$. Note that $\bar{k}^*$ is the key-term with the current highest UCB value and may not be $k^*$.
In order to derive the regret bound, we assume that $a^*_{k^*} = \argmax_a \mathbb{E}[r_{a,t}]$, i.e., the best item in the key-term with the highest reward also has the highest reward among all items. This assumption is consistent with our observation in Section~\ref{sec:key-term} that the key-term rewards are mainly affected by the rewards of the best associated items.
It is also quite general compared to the first-order dominance assumption used to design previous hierarchical bandit algorithms~\cite{jedor2019categorized}, which requires the $i^{th}$ best item in key-term $k^*$ has larger reward than the $i^{th}$ best item of key-term $k$, for any $i,k$.
There are four cases in total that can cause regret in a given round:
\begin{enumerate}
    \item $\bar{k}^*$ is sub-optimal and switching condition is not satisfied
    \item $\bar{k}^*$ is sub-optimal and switching condition is satisfied
    \item $\bar{k}^*$ is optimal and switching condition is not satisfied
    \item $\bar{k}^*$ is optimal and switching condition is satisfied
\end{enumerate}
% \carlee{maybe explain here why you don't separately consider suboptimal $\bar{a}^*$?}
We consider the regret decomposed into these four cases separately.
\subsubsection{Sub-optimal $\bar{k}^*$, no switching}
For sub-optimal key-term $\bar{k}^*$, when the switching condition is not satisfied, we recommend key-term $\bar{k}^*$ and item $\bar{a}^*$. In this case, we only need to consider that 
%\carlee{I am a bit confused here. Shouldn't this inequality be reversed if $k$ is the true optimal key-term and $\bar{k}^*$ the key term that you recommended?}\jinhang{We want to count all bad events that UCB of $k$ is larger than UCB of $k^*$.}
\begin{equation}\label{eq:key-term_UCB}
    \exists \bar{k}^*\neq k^*, \tilde{\hat{\mu}}_{\bar{k}^*} +\tilde{\rho}_{\bar{k}^*} \ge \tilde{\hat{\mu}}_{k^*} +\tilde{\rho}_{k^*}.
\end{equation}
%\carlee{is the idea that $k = \bar{k}^*$ in this inequality? it is hard to keep track of $k^*$ vs. $\bar{k}^*$ vs. $k$ in these explanations}\jinhang{Yes,$k = \bar{k}^*$. I replace all $k$ with $\bar{k}^*$ then.}
By applying Hoeffding’s inequality on $\tilde{\hat\mu}_{\bar{k}^*}$, we have $\tilde{\mu}_{\bar{k}^*} -\tilde{\rho}_{\bar{k}^*} \le \tilde{\hat\mu}_{\bar{k}^*} \le \tilde{\mu}_{\bar{k}^*} + \tilde{\rho}_{\bar{k}^*}$ with high probability (at least $1 - 2|\mathcal{K}|t^{-2}$). Then \eqref{eq:key-term_UCB} becomes
\begin{equation}
    \exists \bar{k}^*\neq k^*, 2\tilde{\rho}_{\bar{k}^*} \geq \tilde{\mu}_{k^*}-\tilde{\mu}_{\bar{k}^*}.
\end{equation}
Substituting $\tilde{\rho}_{\bar{k}^*}$ with its definition $\tilde{\rho}_{\bar{k}^*}=\sqrt{\frac{3 \ln t}{2\tilde{T}_{\bar{k}^*}}}$, we have
\begin{equation}
    \mathbb{E}[\tilde{T}_{\bar{k}^*}] \le\frac{6\ln t}{(\tilde{\mu}_{k^*}-\tilde{\mu}_{\bar{k}^*})^2}.
\end{equation}
The regret caused by selecting sub-optimal key-terms and items is bounded as
\begin{equation}
    Reg_1(T) \leq \sum_{k\neq k^*} \mathbb{E}[\tilde{T}_k] \cdot (\mu_{a^*} - \tilde{\mu}_{k} + 1).
    % \le O\left(\frac{|\mathcal{K}| \ln T}{(\tilde{\mu}_{k^*}-\tilde{\mu}_{k})^2}\right)
\end{equation}

\subsubsection{Sub-optimal $\bar{k}^*$, switching}
By applying Hoeffding’s inequality on $\tilde{\hat\mu}_k$ and $\hat\mu_a$, we have $\tilde{\mu}_k -\tilde{\rho}_k \le \tilde{\hat\mu}_k \le \tilde{\mu}_k + \tilde{\rho}_k$ and $\mu_a -\rho_a \le \hat\mu_a \le \mu_a + \rho_a$ with high probability. If the switching condition is satisfied, we have
\begin{equation}
    \mu_{\bar{a}^*} + (1-\gamma) \cdot \rho_{\bar{a}^*} \ge \tilde{\mu}_{\bar{k}^*} + (\gamma-1) \cdot \tilde{\rho}_{\bar{k}^*},
\end{equation}
\begin{equation}\label{eq:gamma_1}
    \rho_{\bar{a}^*}+\tilde{\rho}_{\bar{k}^*} \leq \frac{1}{\gamma-1} (\mu_{\bar{a}^*} - \tilde{\mu}_{\bar{k}^*}).
\end{equation}
We also have $T_{\bar{a}^*}\le\tilde{T}_{\bar{k}^*}\le\frac{6\ln t}{(\tilde{\mu}_{k^*}-\tilde{\mu}_{\bar{k}^*})^2}$ at the first time that the switching condition is satisfied, which gives us
\begin{equation}\label{eq:gamma_2}
    \rho_{\bar{a}^*}+\tilde{\rho}_{\bar{k}^*} \ge \tilde{\mu}_{k^*}-\tilde{\mu}_{\bar{k}^*}.
\end{equation}
If $\gamma > \frac{\mu_{\bar{a}^*} - \tilde{\mu}_{\bar{k}^*}}{\tilde{\mu}_{k^*}-\tilde{\mu}_{\bar{k}^*}} + 1$, \eqref{eq:gamma_1} and \eqref{eq:gamma_2} will never be true at the same time. Thus, by taking $\gamma \ge \max_{k\neq k^*} \frac{\mu_{a_k^*} - \tilde{\mu}_{k}}{\tilde{\mu}_{k^*}-\tilde{\mu}_{k}} + 1$, the switching condition will never be satisfied and the regret caused by this part is
\begin{equation}
    Reg_2(T) \le O(1).
\end{equation}

\subsubsection{Optimal $\bar{k}^*$, no switching.}
Similar to the previous case, when the switching condition is not satisfied, we have
\begin{equation}
    \mu_{\bar{a}^*} - (\gamma+1) \cdot \rho_{\bar{a}^*} < \tilde{\mu}_{\bar{k}^*} + (\gamma+1) \cdot \tilde{\rho}_{\bar{k}^*},
\end{equation}
\begin{equation}\label{eq:rho_a}
    \rho_{\bar{a}^*}+\tilde{\rho}_{\bar{k}^*} > \frac{1}{\gamma+1} (\mu_{\bar{a}^*} - \tilde{\mu}_{\bar{k}^*}).
\end{equation}
We then bound $\mathbb{E}[\tilde{T}_{k^*}]$ with Lemma 1.
\begin{lemma}\label{lemma:T_k}
\begin{equation}
    \mathbb{E}[\tilde{T}_{k^*}] \le \frac{6(\gamma+1)^2\ln t}{(\mu_{a^*}-\tilde{\mu}_{k^*})^2} + \sum_{a\in\mathcal{A}_{k^*}, a\neq a^*}\frac{6\ln t}{(\mu_{a^*}-\mu_{a})^2}. 
\end{equation}
\end{lemma}
\begin{proof}
For the sub-optimal key-term $\bar{k}^*$, before the switching condition is satisfied, we have
\begin{equation}\label{eq:T_a*}
    T_{a_{\bar{k}^*}^*} \ge \tilde{T}_{\bar{k}^*} - 
    \sum_{a\in\mathcal{A}_{\bar{k}^*}, a\neq a_{\bar{k}^*}^*} T_{a} \ge \tilde{T}_{\bar{k}^*} - \sum_{a\in\mathcal{A}_{k^*}, a\neq a^*}\frac{6\ln t}{(\mu_{a^*}-\mu_{a})^2}.
\end{equation}
Let us consider a threshold  for $\tilde{T}_{\bar{k}^*}$ and discuss two cases.
If $\tilde{T}_{\bar{k}^*} > \frac{6(\gamma+1)^2\ln t}{(\mu_{a^*}-\tilde{\mu}_{k^*})^2} + \sum_{a\in\mathcal{A}_{k^*}, a\neq a^*}\frac{6\ln t}{(\mu_{a^*}-\mu_{a})^2}$, with Eq.\eqref{eq:T_a*}, we have
\begin{equation}
    T_{a_{\bar{k}^*}^*} \ge \frac{6(\gamma+1)^2\ln t}{(\mu_{a^*}-\tilde{\mu}_{k^*})^2}.
\end{equation}
From Eq.\eqref{eq:rho_a}, we can also derive an upper bound for $T_{a_{\bar{k}^*}^*}$:
\begin{equation}\label{eq:T_a_upper}
    T_{a_{\bar{k}^*}^*} < \left(\frac{\mu_{\bar{a}^*} - \tilde{\mu}_{\bar{k}^*}}{(\gamma+1)\sqrt{3\ln t/2}} - \frac{1}{\sqrt{\tilde{T}_{\bar{k}^*}}} \right)^{-2}.
\end{equation}
However, since $\tilde{T}_{\bar{k}^*} > \frac{6(\gamma+1)^2\ln t}{(\mu_{a^*}-\tilde{\mu}_{k^*})^2} + \sum_{a\in\mathcal{A}_{k^*}, a\neq a^*}\frac{6\ln t}{(\mu_{a^*}-\mu_{a})^2}$, Eq.\eqref{eq:T_a_upper} becomes
\begin{equation}
    T_{a_{\bar{k}^*}^*} < \frac{6(\gamma+1)^2\ln t}{(\mu_{a^*}-\tilde{\mu}_{k^*})^2} + \sum_{a\in\mathcal{A}_{k^*}, a\neq a^*}\frac{6\ln t}{(\mu_{a^*}-\mu_{a})^2},
\end{equation}
which is contrary to Eq.\eqref{eq:T_a*}. Thus, we only need to consider the case that $\tilde{T}_{\bar{k}^*} \le \frac{6(\gamma+1)^2\ln t}{(\mu_{a^*}-\tilde{\mu}_{k^*})^2} + \sum_{a\in\mathcal{A}_{k^*}, a\neq a^*}\frac{6\ln t}{(\mu_{a^*}-\mu_{a})^2}$, which provides the upper bound
\begin{equation}
    \mathbb{E}[\tilde{T}_{k^*}] \le \frac{6(\gamma+1)^2\ln t}{(\mu_{a^*}-\tilde{\mu}_{k^*})^2} + \sum_{a\in\mathcal{A}_{k^*}, a\neq a^*}\frac{6\ln t}{(\mu_{a^*}-\mu_{a})^2}. 
\end{equation}
\end{proof}
The regret caused by this part is bounded as
\begin{equation}
    Reg_3(T) \leq \mathbb{E}[\tilde{T}_{k^*}] \cdot (\mu_{a^*} - \tilde{\mu}_{k^*} + 1).
\end{equation}

\subsubsection{Optimal $\bar{k}^*$, switching}
This part can be considered as the regret of the traditional UCB algorithm on the arm set $\mathcal{A}_{k^*}$. 
\begin{equation}
    Reg_4(T) \le \sum_{a\in\mathcal{A}_{k^*},a\neq a^*} \frac{6 \ln t}{\mu_{a^*} - \mu_{a}}.
\end{equation}

\subsubsection{Overall regret}
We define $\Delta_{a} = \mu_{a^*} - \mu_{a}, \tilde{\Delta}_{k} = \tilde{\mu}_{k^*} - \tilde{\mu}_{k}, \Delta^*_{k} = \mu_{a_{k}^*} - \tilde{\mu}_{k}$. The overall regret of Hier-UCB is given below.
\begin{theorem}\label{thm:hierUCB}
Assume that $\gamma > \max_{k\neq k^*} \frac{\Delta^*_{k}}{\tilde{\Delta}_k} + 1$, Hier-UCB has the following regret upper bound
\begin{align}
    Reg(T)=&\sum_{i=1}^4 Reg_i(T)\nonumber\\
    \le& 6\Bigg[\sum_{k\neq k^*} \frac{\Delta^*_{k} + 1}{\tilde{\Delta}_{k}^2} + \sum_{a\in\mathcal{A}_{k^*}, a\neq a^*}\frac{\Delta^*_{k^*}+\Delta_a+1}{\Delta_a^2} \nonumber\\
    &+ \sum_{a\in\mathcal{A}_{k^*}, a\neq a^*}\frac{(\gamma+1)^2(\Delta^*_{k^*}+1)}{(\Delta^*_{k^*})^2}\Bigg]\ln T.
\end{align}
\end{theorem}
{\bf Remark}. Looking at the above distribution dependent
bound, we have the $O\left((|\mathcal{K}| + |\mathcal{A}_{k^*}|)\ln T\right)$ regret, which is asymptotically tight in $T$. Since applying traditional UCB algorithms to item set $\mathcal{A}$ leads to a regret of $O\left(|\mathcal{A}|\ln T\right)$, we improve the regret bound by reducing $|\mathcal{A}|\ln T$ to $(|\mathcal{K}| + |\mathcal{A}_{k^*}|)\ln T$, owing to the utilization of hierarchical structure between key-terms and items. Notice that ConUCB~\cite{zhang2020conversational} in the stochastic setting will lead to the regret of $O\left(|\mathcal{A}|+\ln T\right)$ even without considering the regret caused by key-term asking. Thus, our proposed algorithms prevail over the others when the size of the item pool $|\mathcal{A}|$ is substantially large.
It is possible to analyze the regret of Hier-LinUCB following similar steps of regret decomposition for Hier-UCB, but it is more involved due to users' feature estimation and we left it for the future work.

%% file: 4_experiment.tex
\section{Experiments}\label{sec:exp}
We conduct experiments on both synthetic and real-world data to validate our proposed algorithms.
% We focus on evaluating the accumulated regret and averaged reward in these experiments, and also conduct a case study to validate the unique advantage of Hier-LinUCB over other algorithms.

\subsection{Experimental Setup}
In this section, we discuss our experimental setup. We first introduce two baseline algorithms, then describe the metrics used for evaluation. We also list three research questions that we would like to answer.
\subsubsection{Baselines}
%We choose the following two algorithms as baselines to be compared with.
We compare our algorithm with two baseline algorithms that do not properly model the relationship between user preferences to key-terms and items.
\begin{itemize}
    \item {\bf LinUCB}~\cite{li2010contextual}: A standard non-conversational contextual  bandit algorithm. This comparison demonstrates the benefits of Hier-LinUCB's ability to select key-terms.
    \item {\bf ConUCB}~\cite{zhang2020conversational}: A recent conversational bandit algorithm proposed in the contextual setting that employs a fixed frequency of choosing key-terms. This comparison demonstrates Hier-LinUCB's ability to adaptively select key-terms or items depending on each user, accounting for the cost of choosing key-terms.
\end{itemize}

\subsubsection{Metrics}
We use the cumulative regret in Eq.\eqref{eq:regret} to measure the performance of algorithms. 
As a common metric for bandit problems, it captures the gap between the reward of always choosing the optimal action and the reward of the action chosen by a specific algorithm.
Besides, in Section~\ref{sec:exp_real}, we show the average reward over iterations, $\frac{1}{N}\sum_{t=1}^{N} R_{a_t,k_t,t}$.

\subsubsection{Research Questions}
To validate our approach, we design the following research questions:
\begin{enumerate}
    \item[RQ1.] In a real-world setting, what relationship can be observed between user preferences to key-terms and items?
    \item[RQ2.] Can our algorithms achieve less regret than baseline algorithms by leveraging the relationship and the hierarchical structure between key-terms and items? 
    \item[RQ3.] Can our algorithms adapt to switch between key-term and item recommendations more flexibly than previous work?
\end{enumerate}

The extensive study in Section~\ref{sec:key-term} can answer RQ1. 
We observe that users' preference on a specific key-term is mainly affected by their preference on its representative items. 

\subsection{Synthetic Data}
We consider a stochastic conversational bandit setting with 10 key-terms and 100 items.
We create an item pool $\mathcal{A}=\{a_1,a_2,\cdots,a_{100}\}$ and a key-term pool $\mathcal{K}=\{k_1,k_2,\cdots,k_{10}\}$.  We assume each item is only associated with one key-term, and let $a_1, a_2, \cdots, a_{10}$ associate with $k_1$, $a_{11}, a_{12}, \cdots, a_{20}$ associate with $k_2$, etc. We set each item $a_i$'s expected reward as $\mu_{a_i} = i/100$. 
We assume that the reward of key-term $k$ is equal to the discounted reward of the best associated item, i.e., $\mathbb{E}[\tilde{r}_{k,t}]= \lambda \cdot \max_{a}\left( W_{a,k}\mathbb{E}[r_{a,t}]\right)$, where $\lambda\le1$ is a discount factor. Note that it is a special form of the general assumption required by Hier-UCB in Section~\ref{sec:regret}. We choose $\lambda=0.5$, then for each key-term $k_i$, $\tilde{\mu}_{k_i} = i/2$. The rewards of items and key-terms are generated from Bernoulli distributions. We run Hier-UCB, Hier-LinUCB and traditional UCB algorithms on these items and key-terms. We set $\gamma=1$ for both Hier-UCB and Hier-LinUCB, and $\alpha_t=1$ for Hier-LinUCB. For Hier-LinUCB, as it is designed for the contextual setting, we consider 100-dimensional one-hot contextual vectors for all items, and the algorithm will try to estimate a 100-dimensional feature vector. We repeat the experiment 50 times and show the average cumulative regret with 95\% confidence interval.

\paragraph{Answer to RQ2}
Figure~\ref{fig:reg_hierUCB} shows that Hier-UCB achieves less regret than traditional UCB and Hier-LinUCB in the stochastic setting, owing to the strategic balance between key-term and item recommendations: it spends around 2000 iterations on key-term learning, then switch to explore on a smaller item pool, thus achieves much less regret than UCB. This also validates our regret bound. As Hier-UCB is specifically designed for the stochastic setting with moderate-size item pools, it also outperforms Hier-LinUCB, which is designed for the contextual setting: the least square estimation of feature vectors in Hier-LinUCB may incur higher regrets than directly estimating the expected reward of each arm.

\begin{figure}
    \centering
    \includegraphics[width=0.3\textwidth]{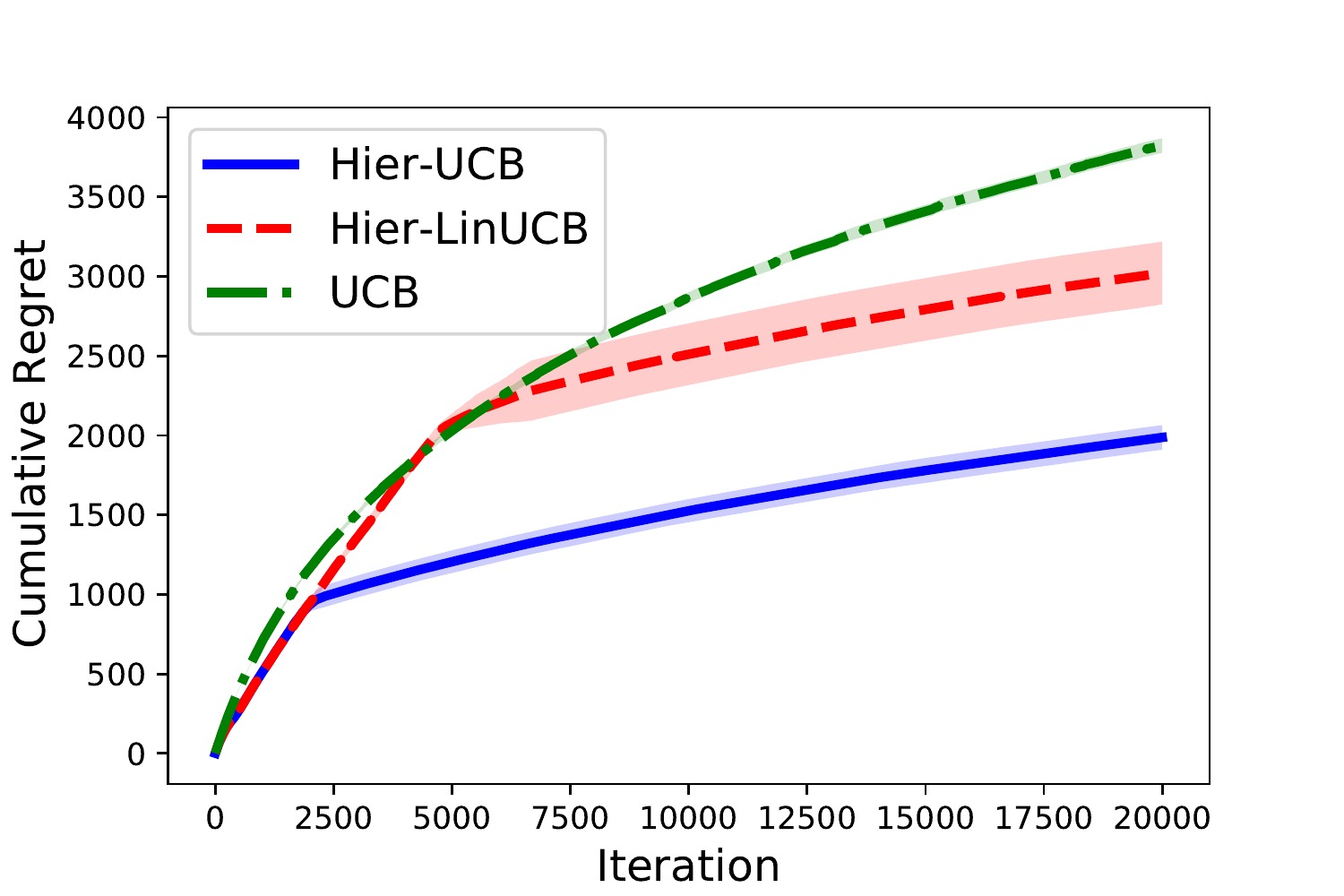}
    \caption{The result validates the regret bound of our algorithm in the stochastic setting.}
    \label{fig:reg_hierUCB}
\end{figure}

\subsection{Real-world Datasets}\label{sec:exp_real}
\subsubsection{Dataset Description}\label{sec:dataset}
In this section, we evaluate Hier-LinUCB on three real-world datasets, and compare its performance with LinUCB and ConUCB. The original datasets are released by Yelp\footnote{https://www.yelp.com/dataset/}, MovieLens\footnote{https://www.grouplens.com}, and LastFM\footnote{https://www.lastfm.com}.
In our experiments, we use the public source data prepossessed by \cite{zhang2020conversational} for the Yelp dataset and \cite{xie2021comparison} for the MovieLens and LastFM datasets.
The Yelp dataset has 1000 key-terms, 5000 items and 200 users, where the key-terms are restaurant categories (e.g., “Burgers”) and the items are restaurants.
The LastFM dataset has 2726 key-terms, 2000 items and 100 users, where the key-terms are artist tags (e.g., “Rock”) and the items are artists.
The Movielens dataset has 5585 key-terms, 2000 items and 200 users, where the key-terms are movie tags (e.g., “Action”) and the items are movies.
We aim to show our algorithms perform well on the Yelp and the LastFM datasets and that, although Movielens has more key-terms than items, which may reduce the benefit of key-term asking, our proposed algorithm can still achieve comparable performance to the previous baselines.
All datasets provide the hierarchical relationship between key-terms and items, the contextual vectors of key-terms and items, and the feature vectors of all users (detailed data generation methods can be found in Section 5.1 of \cite{zhang2020conversational} and Section 4.3 of \cite{xie2021comparison}).

% \textit{Yelp: 200 users, 1000 key-terms and 5000 items. Last FM: 100 users, 2726 key-terms, 2000 items. Movielens: 100 users, 5585 key-terms, 2000 items.}

% We follow ConUCB setting \cite{zhang2020conversational} for the Yelp dataset, and follow Relative ConUCB setting \cite{xie2021comparison} for Last FM and MovieLens dataset. 
For Hier-LinUCB, we set $\gamma=0.5$, $\alpha_{t}=1$ on Yelp and $\gamma=0.5$, $\alpha_{t}=0.25$ on LastFM and MovieLens. For LinUCB and ConUCB, we generally follow the original experimental settings in~\cite{zhang2020conversational,xie2021comparison}, while some experimental configurations are changed to fit our new scenario, which will be discussed in the next section. We choose $b(t) = 10 \lfloor \log(t) \rfloor$ as the conversation frequency function for ConUCB, which leads to the smallest regret in their experiments.

\begin{figure*}[t]
	\centering
	\begin{subfigure}[b]{0.32\textwidth}
		\centering
		\includegraphics[width=\textwidth,trim={0 1cm 0 0},clip]{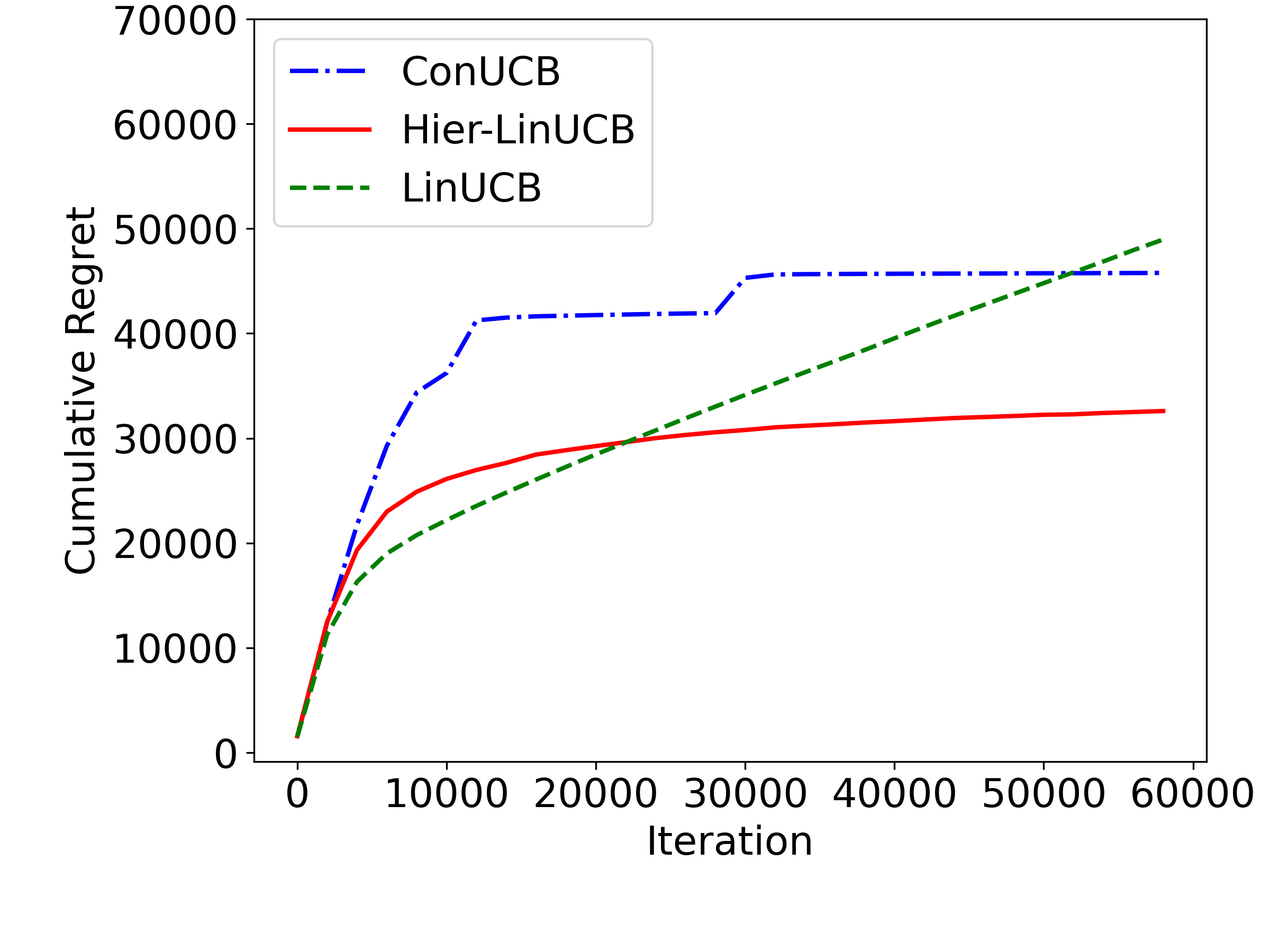}
		\caption{Yelp}
		\label{fig:Yelp_regret}
	\end{subfigure}
	\begin{subfigure}[b]{0.32\textwidth}
    	\centering
    	\includegraphics[width=\textwidth,trim={0 1cm 0 0},clip]{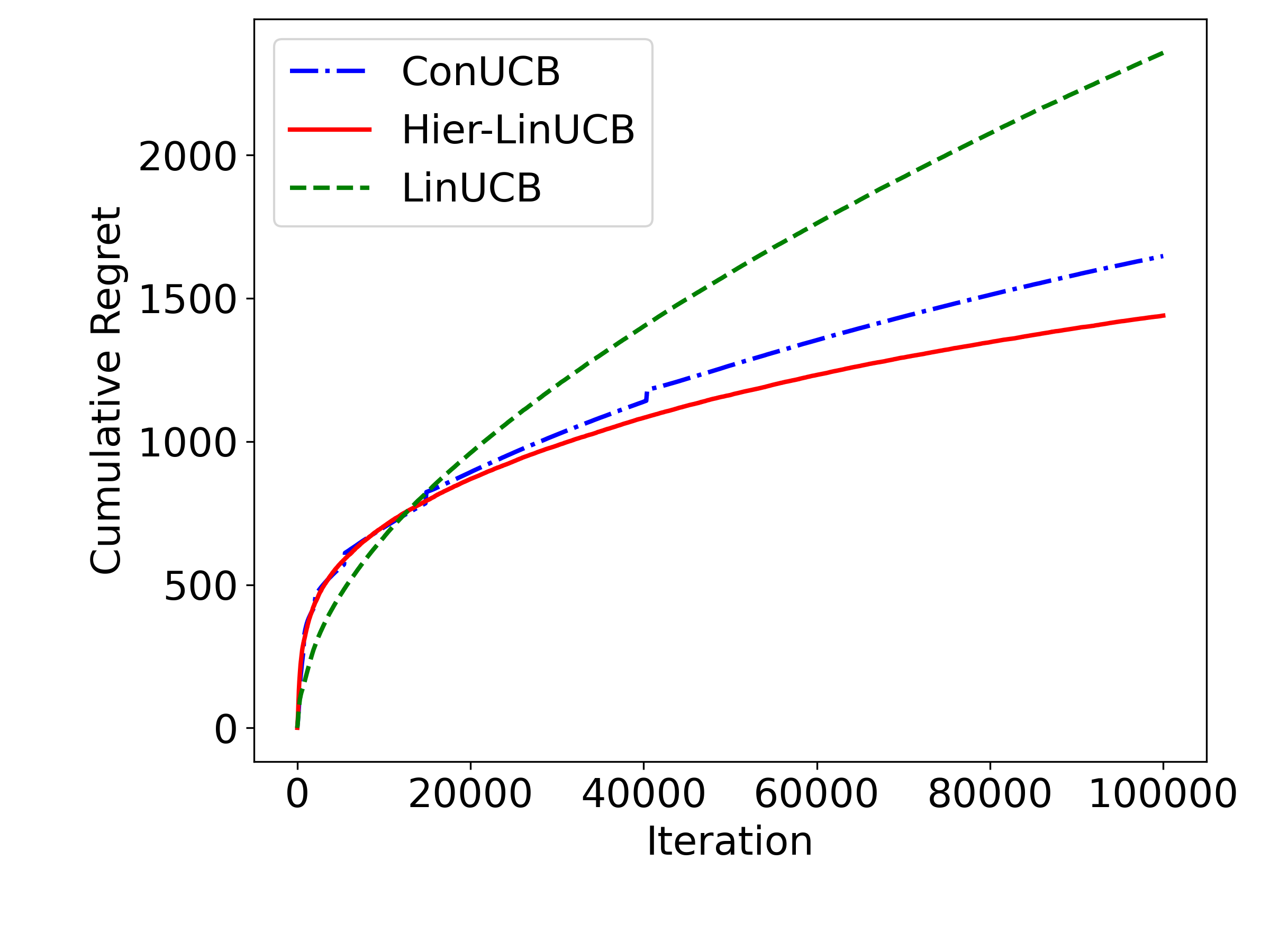}
    	\caption{LastFM}
    	\label{fig:LastFM_regret}
	\end{subfigure}
	\begin{subfigure}[b]{0.32\textwidth}
	\centering
	\includegraphics[width=\textwidth,trim={0 1cm 0 0},clip]{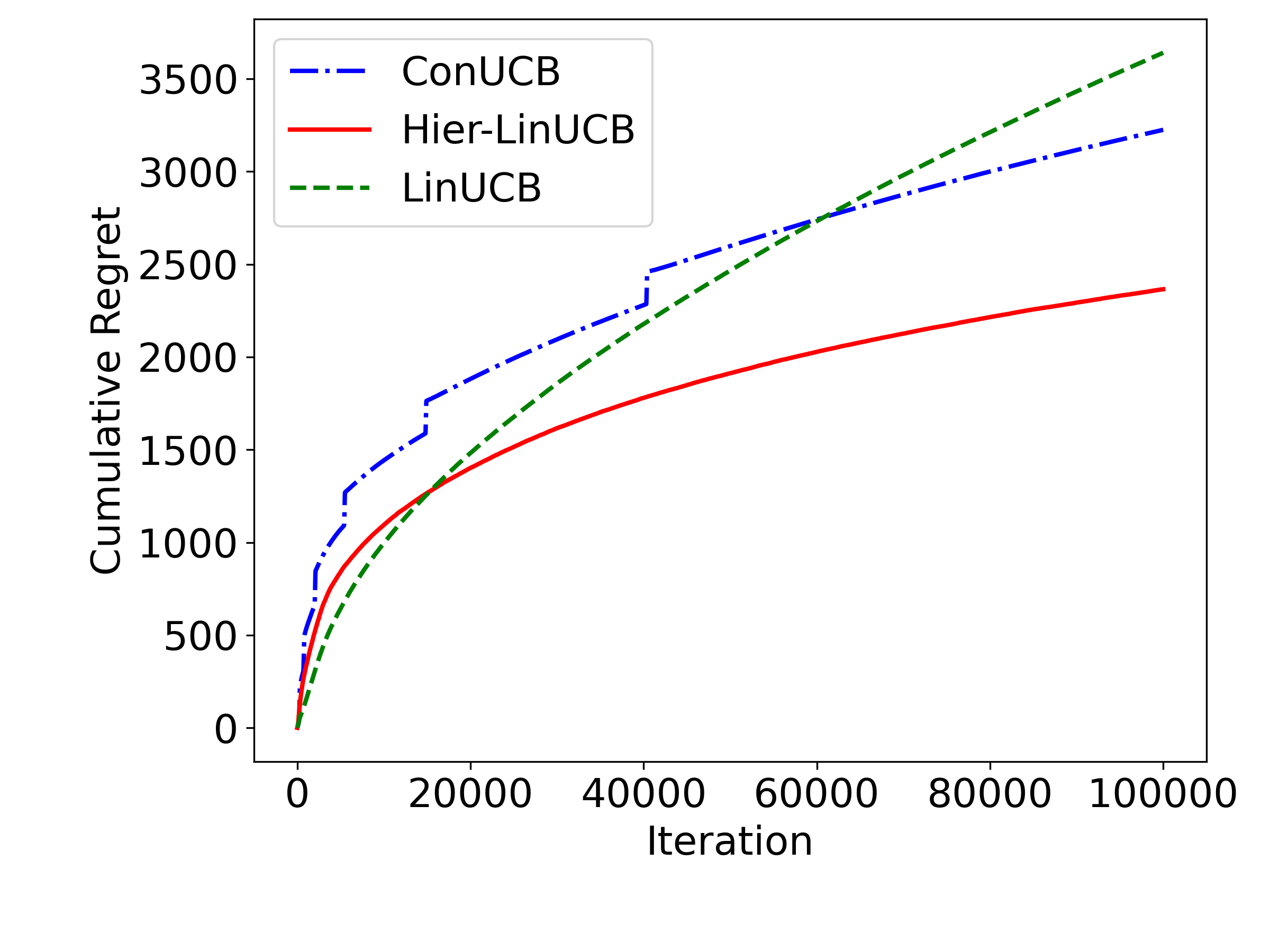}
	\caption{MovieLens}
	\label{fig:MovieLens_regret}
\end{subfigure}
	\caption{Comparison of the cumulative regrets on real-world datasets using Hier-LinUCB, LinUCB and ConUCB. Hier-LinUCB has the lowest converged regret after more than 30000 iterations. The improvement is most apparent on Yelp, likely because the MovieLens and LastFM datasets have a large number of key terms, requiring more key-term exploration.}
	\label{fig:all_regret}
\end{figure*}

\begin{figure*}[t]
	\centering
	\begin{subfigure}[b]{0.32\textwidth}
		\centering
		\includegraphics[width=\textwidth,trim={0 1cm 0 0},clip]{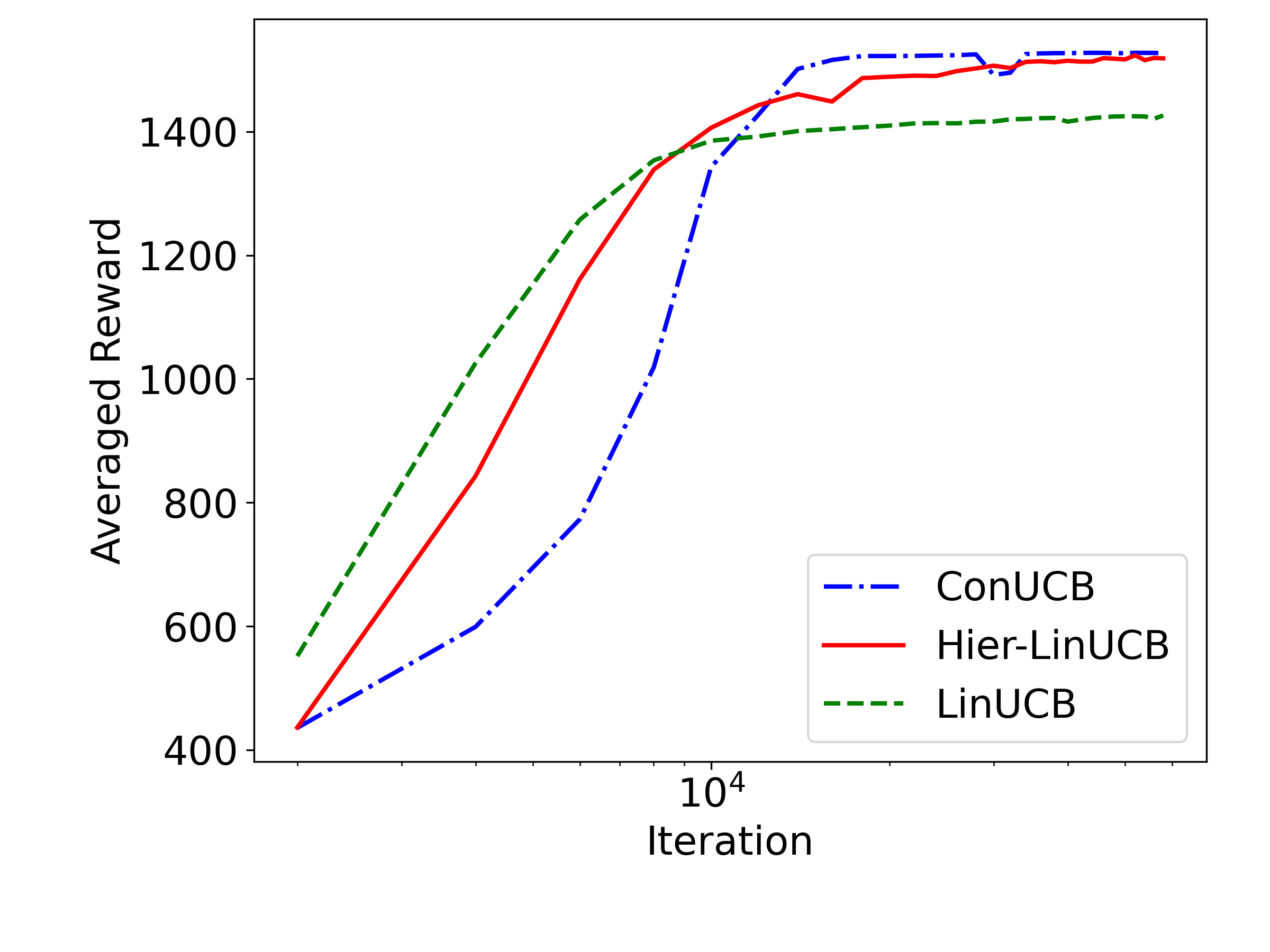}
		\caption{Yelp}
		\label{fig:Yelp_reward}
	\end{subfigure}
	\begin{subfigure}[b]{0.32\textwidth}
    	\centering
    	\includegraphics[width=\textwidth,trim={0 1cm 0 0},clip]{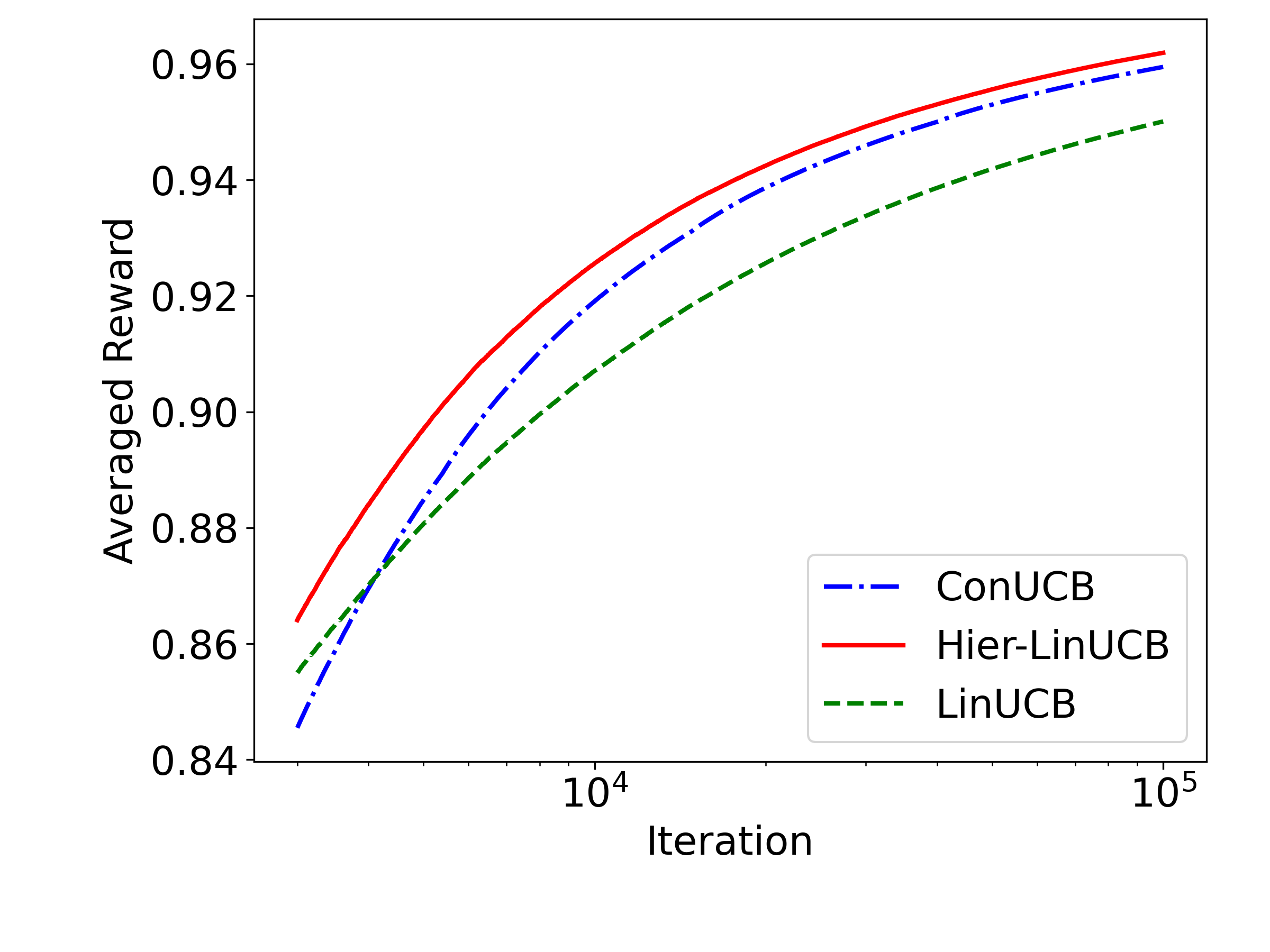}
    	\caption{LastFM}
    	\label{fig:LastFM_reward}
	\end{subfigure}
	\begin{subfigure}[b]{0.32\textwidth}
		\centering
		\includegraphics[width=\textwidth,trim={0 1cm 0 0},clip]{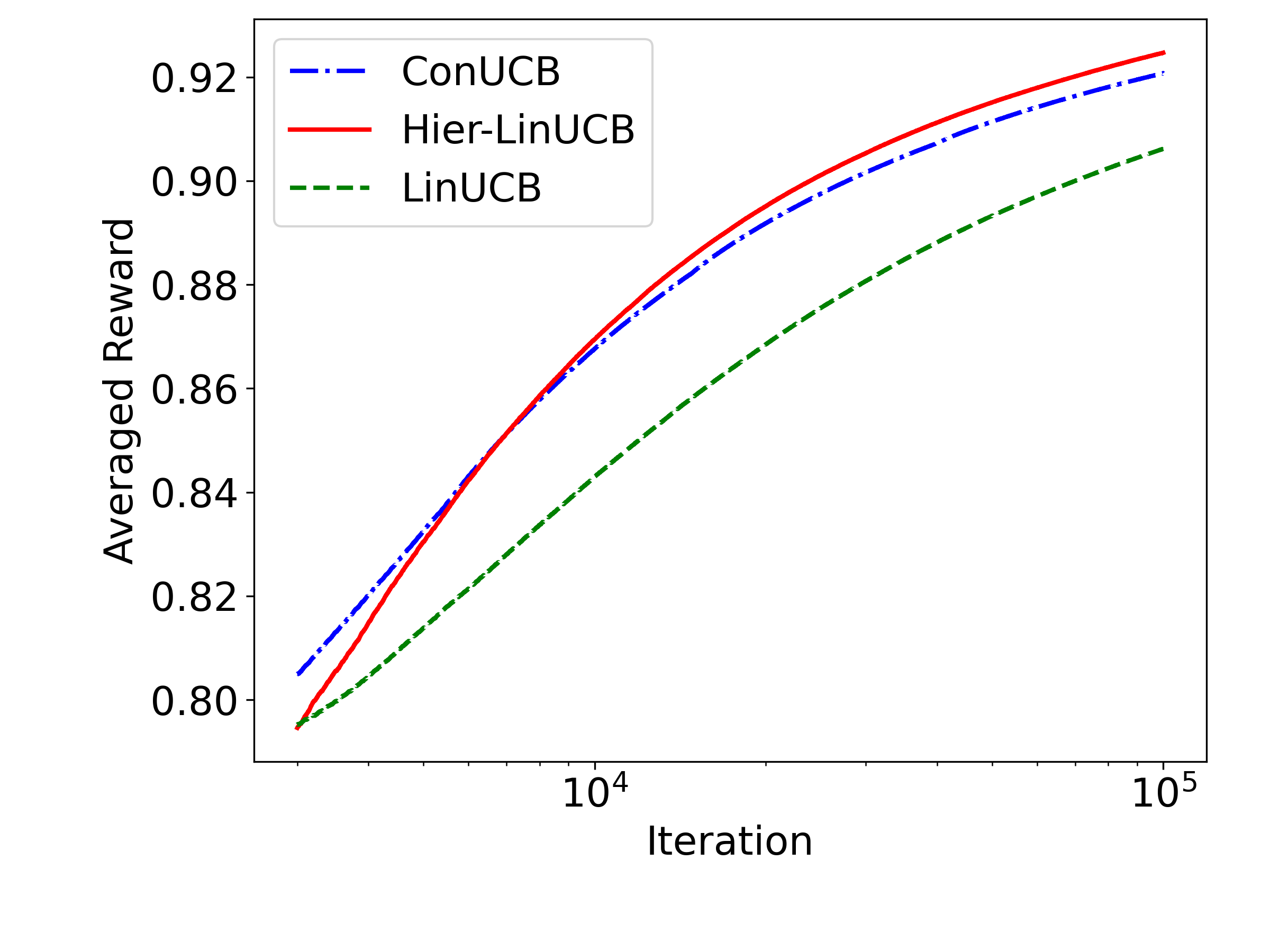}
		\caption{MovieLens}
		\label{fig:MovieLens_reward}
	\end{subfigure}
	\caption{Comparison of the averaged rewards on real-world datasets using Hier-LinUCB, LinUCB and ConUCB. The averaged reward of Hier-LinUCB grows faster than ConUCB on Yelp and LastFM, while slightly slower than ConUCB on MovieLens at the early stage. This is likely due to the large number of key-terms for MovieLens.}
	\label{fig:all_reward}
\end{figure*}

\subsubsection{Experimental Settings}
\paragraph{Key-term reward}
We assume that the reward of key-term $k$ is equal to the discounted reward of the best associated item, i.e., $\mathbb{E}[\tilde{r}_{k,t}]= \lambda \cdot \max_{a}\left( W_{a,k}\mathbb{E}[r_{a,t}]\right)$, where $\lambda\le1$ is a discount factor. It is a special form of the general assumption in Section~\ref{sec:regret}.
We set $\lambda = 0.5$ and re-calculate the expected rewards and feature vectors of key-terms.
% To correspond to the change of key-term reward, we change feature vector accordingly.
% so that the key-term reward is determined by the reward of the best associated item.
% For the reward of key-terms, we choose the best reward from the related item, and multiple it with discount. 
\paragraph{Item pool size} ConUCB~\cite{zhang2020conversational} randomly chooses an arm pool of 50 items at each round in their original experiments. However, such an arm pool is too small compared to the whole arm pool with thousands of items, and it is very likely to exclude many good items with high rewards, 
Instead, we use the full item pool in our experiments, which also makes the learning problem more challenging,
To encourage exploration, we set $\alpha_{t}=1$ for all baseline algorithms. 
We observe higher cumulative regrets than the original experiments, which is consistent with Figure 2 in~\cite{zhang2020conversational}.
\subsubsection{Results}
In our experiments, we compare the performance of Hier-LinUCB with LinUCB and ConUCB. We also conduct a case study of two users in the Yelp dataset with ConUCB.
\paragraph{Answer to RQ2}
We show the cumulative regrets and averaged rewards of different algorithms in Figures~\ref{fig:all_regret} and \ref{fig:all_reward}. Notice that the system recommends one key-term or item to one user at each interaction; when the averaged reward almost converges after 20000 interactions for Yelp, the average number of key-term asking and item recommendations per user is 100, which is reasonable after a few days or weeks.
Figure~\ref{fig:all_regret} shows that Hier-LinUCB achieves lower regrets than LinUCB and ConUCB after iteration ends.
At the very beginning, due to more exploration on key-terms, Hier-LinUCB may not achieve lower regrets than LinUCB and ConUCB. However, as long as it learns the correct key-terms, it will converge to the optimal actions much faster than others. We observe that the regret of ConUCB periodically increases, which comes from key term conversations predefined by the conversation frequency function~\cite{zhang2020conversational}.
To better observe the early state rewards and understand the behavior of different algorithms in cold-start recommender systems, we further compare the averaged reward with the logarithmic x-axis in Figure~\ref{fig:all_reward}.
%We use the logarithmic x-axis in Figure~\ref{fig:all_reward} as we care more about the early state rewards of cold-start recommender systems.
The results are shown in Figure~\ref{fig:all_reward}.
Although LinUCB can achieve good rewards at the early stage, it may not find the optimal item eventually compared to ConUCB and Hier-LinUCB, showing the necessity of key-term asking. Besides, the averaged rewards of ConUCB and Hier-LinUCB converge to the similar levels once they successfully identify the best key-terms. 
However, in Figures~\ref{fig:Yelp_reward} and \ref{fig:LastFM_reward}, Hier-LinUCB searches faster than ConUCB and has higher rewards for the early iterations, since the number of key-terms is similar to (or less than) the number of items for the Yelp and LastFM datasets. When the number of key-terms is larger than the number of items (MovieLens), as shown in Figure~\ref{fig:MovieLens_reward}, Hier-LinUCB has to explore more at first (resulting in initially lower rewards), before eventually achieving higher reward than ConUCB. 
Notice that the former case (\emph{i.e.}, number of key-terms is limited) is common in practice, since the key-terms of the items need to be labeled carefully, which usually leads to a lot of human efforts.
% In Figure~\ref{fig:Yelp_reward}, Hier-LinUCB has more rapid reward growth than ConUCB, which likely due to its adaptive exploration policy of key-terms.
% In Figure~\ref{fig:MovieLens_reward} and Figure~\ref{fig:LastFM_reward}, the difference of reward growth rate is less significant. This probably results from the key-term pool sizes of these two datasets. As introduced in Section~\ref{sec:dataset}, the key-term pool is smaller than the item pool for Yelp, while the key-term pool is larger than the item pool for the MovieLens and LastFM datasets, which makes the exploration of key-terms more difficult. The learning rate of Hier-LinUCB mainly depends on key-terms, so it is natural to have better performance on Yelp than MovieLens and LastFM. 
% Indeed, ConUCB also achieves similar rewards to LinUCB on LastFM, indicating that key-term exploration \emph{in general} may not be useful for finding the best item on this dataset.

\begin{figure}
    \centering
    \includegraphics[width=0.40\textwidth]{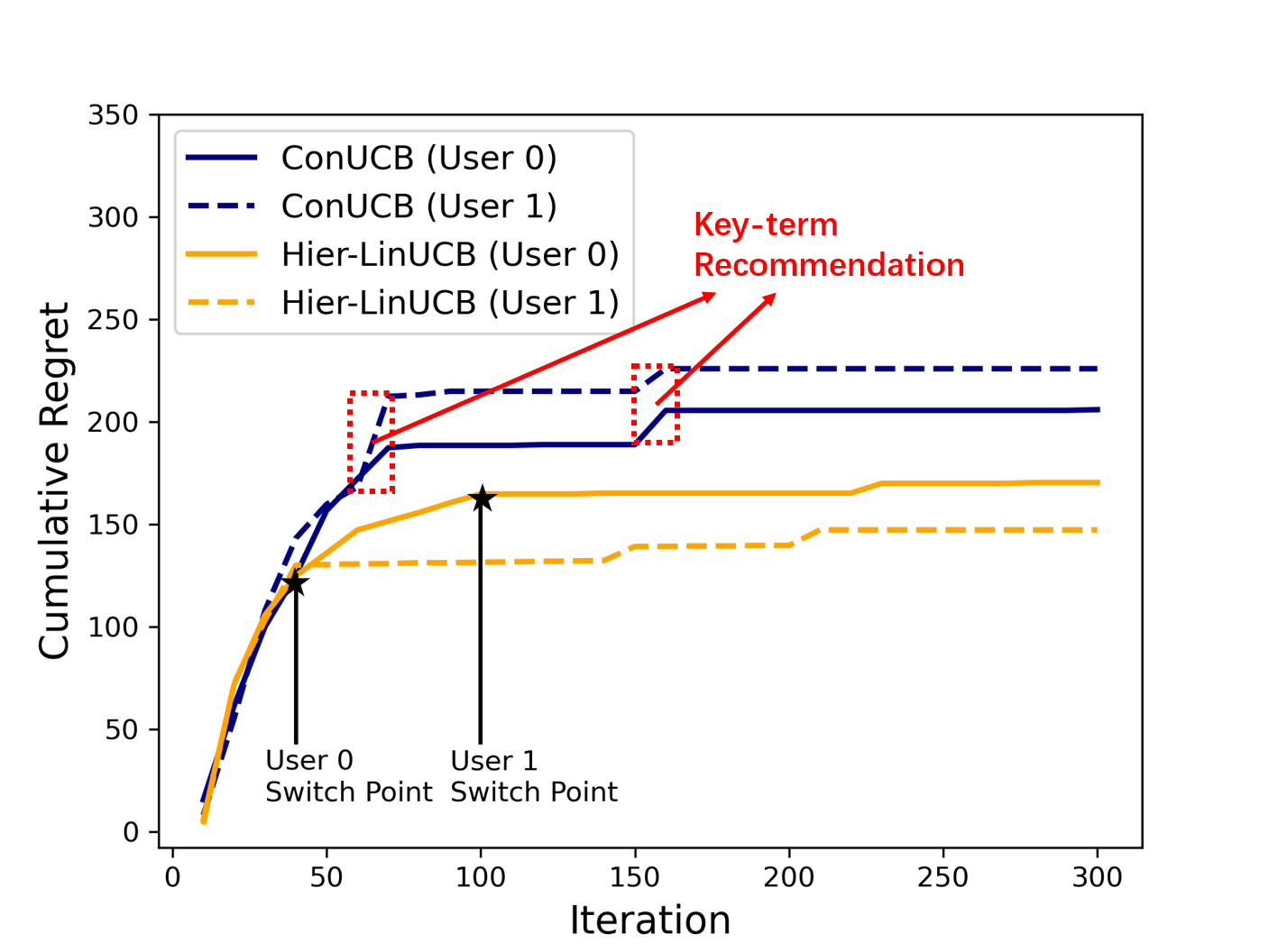}
    \caption{A case study of two users on the Yelp dataset. ConUCB has fixed exploration on key-terms (\emph{i.e.}, switching from item recommendations to key-term asking) for user 0 and user 1 (marked by red color), while Hier-LinUCB enables more flexible and personalized switching between key-term asking and item recommendations.}
    \label{fig:switch}
\end{figure}

\paragraph{Answer to RQ3}
We pick two users from the Yelp dataset and monitor their regrets in Figure~\ref{fig:switch}. We find Hier-LinUCB has different switch points for different users: before the switch point, it explores more on key-terms; after the switch point, it only chooses items and the regret converges quickly. 
By contrast, ConUCB has a fixed key-term conversation period as highlighted in red. Hier-LinUCB thus enables more flexible and personalized switching between key-term asking and item recommendations. As discussed before, the key-term exploration of Hier-LinUCB is based on the switching condition, while the key-term exploration of ConUCB is based on a fixed conversation frequency function. Hier-LinUCB treats User 1 and User 0 differently, while ConUCB uses a fixed exploration policy and results in worse regrets.

%% file: 6_conclusion.tex
\section{Conclusion}\label{sec:conclusion}
To avoid the excessive conversational interactions incurred by the pre-defined conversation frequency of existing conversational bandit algorithms, we formulate a new conversational bandit problem that allows the recommender system to choose either a key-term or an item to recommend at each round.
To balancing the new exploration-exploitation (EE) trade-off between key-term asking and item recommendation, we need accurately understand the relationship between key-term and item rewards.
We study such relationship based on a survey and analysis on a real-world dataset. We observe that key-term rewards are mainly affected by rewards of representative items. We propose two bandit algorithms based on this observation and the hierarchical structure between key-terms and items. We prove the theoretical regret bound of our algorithm and validate them on synthetic and real-world data.

\begin{acks}
This work was supported by ONR Grant N000142112128.
\end{acks}
\clearpage
\balance
% Our work locates a problem of conversational recommendation, pointing out the recommendation of key-terms should be sampled from the best related items rather than simple average of them. Following the problem, we propose Hier-UCB, a bandit algorithm with hierarchical structure based on upper confidence bound. It also features top-based key-term reward and flexible key-term to item exploration switching. To apply the advantage of hierarchical structure and top-based key-term reward, we also proposed Hier-LinUCB, a contextual version of Hier-UCB that takes the user context into account. Three real world dataset experiments validates the advantage of Hier-LinUCB over ConUCB and LinUCB with its lower cumulative regret. A case study also validates the flexible switch in the contextual setting.